\documentclass[aap,preprint]{imsart}
\RequirePackage{amsthm,amsmath}
\RequirePackage{color}
\RequirePackage[numbers]{natbib}
\usepackage{amssymb,amsfonts,bbm}
\usepackage{graphicx}
\usepackage{comment}

\arxiv{arXiv:0000.0000}

\startlocaldefs
\numberwithin{equation}{section}
\theoremstyle{plain}

\newtheorem{dummytheorem}{Dummy-Theorem}[section]
\newcommand{\proofendsign}{$\Box$} 

\newtheorem{lemma}[dummytheorem]{Lemma}
\newtheorem{theorem}[dummytheorem]{Theorem}
\newtheorem{proposition}[dummytheorem]{Proposition}

\newtheorem{corollary}[dummytheorem]{Corollary}

\newtheorem{example}[dummytheorem]{Example}

\newtheorem{remark}[dummytheorem]{Remark}

\renewenvironment{proof}{{\noindent \bf Proof }}
 {{\hspace*{\fill}\proofendsign\par\bigskip}}

\endlocaldefs

\newcommand{\N}{\mathbb{N}}

\newcommand{\Q}{\mathrm{Q}}
\newcommand{\R}{\mathbb{R}}

\newcommand{\E}{\mathbb{E}}

\newcommand{\HHH}{\mathbb{H}}
\newcommand{\TTT}{\mathbb{T}}

\newcommand{\pr}{\mathrm{P}}
\newcommand{\ex}{\mathbb{E}}

\newcommand{\eins}{\mathbbm{1}}

\newcommand{\cA}{\mathcal A}
\newcommand{\cB}{\mathcal B}
\newcommand{\cC}{\mathcal C}
\newcommand{\cF}{\mathcal F}

\newcommand{\cM}{\mathcal M}
\newcommand{\cP}{\mathcal P}
\newcommand{\cQ}{\mathcal Q}

\newcommand{\cT}{\mathcal T}
\newcommand{\cTr}{\mathcal T^{r}}
\newcommand{\cTT}{{\mathcal T}_{\mathbb{T}}}

\newcommand{\cU}{\mathcal U}
\newcommand{\cX}{\mathcal X}

\newcommand{\esssup}{\mathop{\mathrm{ess\,sup}}\displaylimits}
\newcommand{\essinf}{\mathop{\mathrm{ess\,inf}}\displaylimits}

\newcommand{\OFP}{(\Omega,{\cal F},\pr)}
\newcommand{\tOFP}{(\overline{\Omega},\overline{{\cal F}},\overline{\pr})}
\newcommand{\OFPk}{(\Omega,{\cal F}_{t_{k}},\pr_{|{\cal F}_{t_{k}}})}
\newcommand{\OFFP}{(\Omega,{\cal F}, ({\cal F}_{t})_{0\leq t\leq T},\pr)}
\newcommand{\tOFFP}{(\Omega,{\cal F}, (\widetilde{{\cal F}}_{t})_{0\leq t\leq \infty},\pr)}
\newcommand{\oOmega}{\overline{\Omega}}
\newcommand{\ocF}{\overline{{\mathcal F}}}
\newcommand{\oP}{\overline{\pr}}
\newcommand{\ocPinfty}{\overline{{\mathcal P}}^{\infty}}
\newcommand{\taur}{\tau^{r}}

\def\bcswitch{\left\{\renewcommand{\arraystretch}{1.2}\begin{array}{c@{,~}c}}

\def\ecswitch{\end{array}\right.}
\newcommand{\TimestT}{\raisebox{-2.2mm}{
\unitlength 1.00pt
\linethickness{0.5pt}
\begin{picture}(10.00,20.00)
\multiput(2.00,6.00)(0.12,0.18){50}{\line(0,1){0.18}}
\multiput(2.00,15.00)(0.12,-0.18){50}{\line(0,-1){0.18}}
\put(5.00,2.00){\makebox(0,0)[cc]{$\scriptstyle t=1$}}
\put(5.00,18.00){\makebox(0,0)[cc]{$\scriptstyle T$}}
\end{picture}
}}
\newcommand{\Timesim}{\raisebox{-2.2mm}{
\unitlength 1.00pt
\linethickness{0.5pt}
\begin{picture}(10.00,20.00)
\multiput(2.00,6.00)(0.12,0.18){50}{\line(0,1){0.18}}
\multiput(2.00,15.00)(0.12,-0.18){50}{\line(0,-1){0.18}}
\put(5.00,2.00){\makebox(0,0)[cc]{$\scriptstyle i=1$}}
\put(5.00,18.00){\makebox(0,0)[cc]{$\scriptstyle m$}}
\end{picture}
}}
\newcommand{\Timesikm}{\raisebox{-2.2mm}{
\unitlength 1.00pt
\linethickness{0.5pt}
\begin{picture}(10.00,20.00)
\multiput(2.00,6.00)(0.12,0.18){50}{\line(0,1){0.18}}
\multiput(2.00,15.00)(0.12,-0.18){50}{\line(0,-1){0.18}}
\put(5.00,2.00){\makebox(0,0)[cc]{$\scriptstyle i=k$}}
\put(5.00,18.00){\makebox(0,0)[cc]{$\scriptstyle m$}}
\end{picture}
}}

\newcommand{\Timesk}{\raisebox{-2.2mm}{
\unitlength 1.00pt
\linethickness{0.5pt}
\begin{picture}(10.00,20.00)
\multiput(2.00,6.00)(0.12,0.18){50}{\line(0,1){0.18}}
\multiput(2.00,15.00)(0.12,-0.18){50}{\line(0,-1){0.18}}
\put(5.00,2.00){\makebox(0,0)[cc]{$\scriptstyle k=1$}}
\put(5.00,18.00){\makebox(0,0)[cc]{$\scriptstyle k_{j}$}}
\end{picture}
}}

\begin{document}

\begin{frontmatter}
\title{Optimal stopping under  model uncertainty: randomized stopping times approach }
\runtitle{Optimal stopping under uncertainty}

\begin{aug}
\author{\snm{Denis Belomestny}
\ead[label=e1]{denis.belomestny@uni-due.de}}
\and
\author{\snm{Volker Kr\"atschmer}\ead[label=e2]{volker.kraetschmer@uni-due.de}}


\thankstext{T1}{This research was partially supported by the Deutsche
      Forschungsgemeinschaft through the SPP 1324 ``Mathematical methods for extracting quantifiable     information from complex systems'' and   by
Laboratory for Structural Methods of Data Analysis in Predictive Modeling, MIPT, RF government grant, ag. 11.G34.31.0073.}

\runauthor{D. Belomestny and V. Kr\"atschmer}

\affiliation{Duisburg-Essen University}

\address{
Duisburg-Essen University\\
Faculty of Mathematics\\
Thea-Leymann-Str. 9\\
D-45127 Essen\\
Germany\\
\printead{e1}\\
\phantom{E-mail: \ }\printead*{e2}}

\end{aug}

\begin{abstract}
In this work we consider optimal stopping problems with conditional convex risk measures  
of the form
$$
\rho^{\Phi}_t(X)=\sup_{\Q\in\cQ_{t}}\left(\ex_{\Q}[-X|\mathcal{F}_t]- \ex\left[\Phi\left(\frac{d\Q}{d\pr}\right)\big|\cF_{t}\right]\right), 
$$
where $\Phi: [0,\infty[\rightarrow [0,\infty]$ is a lower semicontinuous convex mapping and $\cQ_{t}$ stands for the set of all probability measures $\Q$ which are absolutely continuous w.r.t. a given measure $\pr$   and $\Q= \pr$ on \(\cF_{t}.\) Here the \textit{model uncertainty risk} depends on a (random) divergence 
$\ex\left[\Phi\left(\frac{d\Q}{d\pr}\right)\big|\cF_{t}\right]$ measuring the distance between a hypothetical probability measure we are uncertain about and a reference one at time $t.$
Let \((Y_t)_{t\in [0,T]}\) be an adapted nonnegative, right-continuous stochastic process  fulfilling some proper integrability condition and let \(\cT\) be the set of stopping times on \([0,T]\), then without assuming any kind of time-consistency for the family \((\rho_t^{\Phi}),\) we derive a novel representation 
\begin{eqnarray*}
\sup_{\tau\in\cT}\rho^{\Phi}_0(-Y_\tau)= \inf_{x\in\R}\left\{\sup\limits_{\tau\in\cT}\ex\bigl[\Phi^*(x + Y_{\tau}) - x\bigr]\right\},
\end{eqnarray*}
which makes the application of the standard dynamic programming based approaches possible.
In particular, we generalize the additive dual representation of Rogers, \cite{Rogers2002} to the   case of optimal stopping  under uncertainty. Finally, we develop several Monte Carlo algorithms and illustrate their power for  optimal stopping under Average Value at Risk.
\end{abstract}



\end{frontmatter}
\section{Introduction}
In this paper we study the optimal stopping problems in an uncertain environment. 
The classical solution to the optimal stopping problems based on the dynamic programming principle assumes that there is a unique subjective prior distribution driving the reward process. 
However, for example, in incomplete financial markets, we  have to deal  with multiple equivalent martingale measures not being sure which one underlies the market.    In fact under the presence of the multiple possible distributions, a solution of the optimal stopping problem by maximization with respect to some subjective prior cannot be reliable. Instead, it is reasonable  to view the multitude of possible distributions as a kind of \textit{model uncertainty risk} which should be taken into account while formulating an optimal stopping problem. Here one may draw on concepts from the theory of risk measures. As the established generic notion for static risk assessment at present time $0$, convex risk measures are specific functionals $\rho_{0}$ on vector spaces of random variables viewed as financial risks ({see \cite{FrittelliRosazza2002} and \cite{FrittelliRosazza2004}}). They typically have the following type of robust representation 
\begin{eqnarray}
\label{rhostatic}
\rho_{0}(X)=\sup_{\Q\in \cQ(\pr)}\bigl\{\E_{\Q}[-X] - \gamma_{0}(\Q)\bigr\},
\end{eqnarray}
where $\cQ(\pr)$ denotes the set of probability measures which are absolutely continuous w.r.t. a given reference probability measure $\pr,$ and $\gamma_{0}$ is some penalty function (see e.g. {\cite{CheriditoDelbaen2004}} and \cite{FoellmerSchied2010}). In this way, model uncertainty  is incorporated, as no specific probability measure is assumed. Moreover, the penalty function  scales the plausibility of models. 

Turning over from static to dynamic risk assessment, convex risk measures have been extended to the concept of conditional convex risk measures $\rho_{t}$ at a future time $t,$ which are specific functions on the space of financial risks with random outcomes {(see \cite{BionNadal2008}, \cite{DetlefsenScandolo2005} and \cite{CheriditoDelbaen2006})}. Under some regularity conditions, they have a robust representation of the form (see e.g. \cite{FoellmerPenner2006}, {\cite{DelbaenPeng2010}} or \cite[Chap. 11]{FoellmerSchied2010})
\begin{eqnarray}
\label{rho_conv}
\rho_t(X)=\sup_{\Q\in \mathcal{Q}_t} \bigl\{\E_{\Q}[-X|\mathcal{F}_t]-\gamma_t(\Q)\bigr\},
\end{eqnarray}
where \(\gamma_t\) is a (random) penalty function and \(\mathcal{Q}_t\) consists of all $\Q\in \mathcal{Q}(\pr)$ with $\Q=\pr \mbox{ on } \mathcal{F}_t.$ As in (\ref{rhostatic}), the robust representation \eqref{rho_conv} mirrors the model uncertainty, but now at a future time $t.$  

In recent years the optimal stopping with families $(\rho_{t})_{t\in [0,T]}$ of conditional convex risk measures was subject of several studies. For example, the works 
\cite{Riedel2009} and \cite{KraetschmerSchoenmakers2010} are settled within a time-discrete framework, where in addition the latter one provides some dual representations extending the well-known ones from the classical optimal stopping. Optimal stopping in continuous time was considered in \cite{BayraktarKaratzasYao2010}, \cite{BayraktarYao2011I}, \cite{BayraktarYao2011II}, \cite{ChengRiedel2013}.  All these contributions restrict their analysis to the families $(\rho_{t})_{t\in [0,T]}$ satisfying the property of time consistency, sometimes also called recursiveness, defined to mean
\begin{eqnarray*}
\rho_s(-\rho_{t})=\rho_s,\quad 0\leq s < t\leq T. 
\end{eqnarray*}
Hence the results of the above papers can not be, for example, used to solve optimal stopping problems under such very popular convex risk measure as Average Value at Risk.
The only paper which tackled the case of non time-consistent families of conditional convex risk measures so far is \cite{XuZhou2013}, where the authors considered the so-called  distorted mean payoff functionals.  
However, the analysis of \cite{XuZhou2013} excludes the case of Average Value at Risk as well. Moreover, the class of processes to be stopped is limited to the functions of a one-dimensional geometric Brownian motion. The main probabilistic tool used in \cite{XuZhou2013} is the Skorokhod embedding. 
\par
 In this paper we consider a rather general class of conditional convex risk measures having representation  \eqref{rho_conv}
with  
 $\gamma_{t}(\Q) = \ex\big[\Phi\big(d\Q/d\pr\big)\big|\cF_{t}\big]$  for some lower semicontinuous convex mapping $\Phi: [0,\infty[\rightarrow [0,\infty].$ The related class of risk measures $\rho_{0}$ known as the class of \textit{divergence risk measures} or \textit{optimized certainty equivalents}  was first introduced in \cite{Ben-TalTeboulle1987}, 
 \cite{Ben-TalTeboulle2007}. Any divergence risk measure has the representation
 $$
 \rho_{0}(X) = \inf_{x\in\R}\ex\big[\Phi^{*}\big(x - X\big) - x\big]
 $$ 
 with
 $$
 \Phi^{*}: \R\rightarrow [0,\infty],\, y\mapsto\sup_{x\geq 0}(xy - \Phi(x)).
 $$
 (cf. \cite{Ben-TalTeboulle1987},  \cite{Ben-TalTeboulle2007}, \cite{CheriditoLi2009}, or Appendix \ref{AppendixAA}).
Here we study the  problem of optimally stopping the reward process \(\rho_0(-Y_t),\) where \((Y_t)_{t\in [0,T]}\) is an adapted nonnegative, right-continuous stochastic process with \(\sup_{t\in [0,T]} Y_t\)  satisfying some suitable integrability condition. We do not assume any time-consistency  for the family \(\rho_t \) and  basically impose no further restrictions on \((Y_t)\). Our main result is the representation  
\begin{eqnarray}
\label{repr_abs}
\sup_{\tau\in\cT}\rho_0(-Y_\tau)= \inf_{x\in\R}\left\{\sup_{\tau\in\cT}\ex[\Phi^{*}(x + Y_{\tau}) - x]\right\},
\end{eqnarray}
which allows one to apply the well known methods from  the theory of ordinary optimal stopping problems. In particular, we derive the so-called  additive dual representation of the form:
\begin{eqnarray}
\label{dual_intr}
\inf_{x\in\R}\inf_{M\in \cM_{0}}\ex\left[\sup_{t\in [0,T]}\big(\Phi^{*}(x+ Y_{t}) - x - M_{t}\big)\right],
\end{eqnarray}
where \(\cM_{0}\) is the class of adapted martingales vanishing at time 0. This dual representation generalizes the well-known dual representation of Rogers, \cite{Rogers2002}.
The representation \eqref{dual_intr} together with \eqref{repr_abs} can be used to efficiently construct lower and upper bounds for the optimal value \eqref{repr_abs} by Monte Carlo.
\par
The paper is organised as follows. In Section~\ref{setup} we introduce  notation and set up the optimal stopping problem.  The main results are presented in Section~\ref{main_results} where in particular a criterion  ensuring the existence of a saddle-point in \eqref{repr_abs} is formulated.  Section~\ref{discussion} contains some discussion on the main results and on their relation to the previous literature.  A Monte Carlo algorithm for computing lower and upper bounds for the value function is formulated in Section~\ref{MC}, where also an  example of  optimal stopping under Average Value at Risk is numerically analized.\par 
The crucial idea to derive representation \eqref{repr_abs} is to consider the optimal stopping problem
$$
\mbox{maximize}~\rho_{0}(-Y_{\taur})~\mbox{over}~\taur\in\cT^{r},
$$
where $\cT^{r}$ denotes the set of all randomized stopping times on $[0,T].$ It will be studied in Section \ref{optimalrandomizedstoppingtimes}, where in particular it will turn out that this optimal stopping problem has the same optimal value as the originial one. Finally, the proofs are collected in Section~\ref{proofs}.
\section{The set-up}
\label{setup}
Let $(\Omega,\cF, \pr)$ be a probability space and denote by $L^0:=L^0(\Omega,\cF,\pr)$ the class of all finitely-valued random variables (modulo the $\pr$-a.s. equivalence).  Let \(\Psi\) be  a Young function, i.e., a left-continuous, nondecreasing convex function $\Psi:\R_+\to[0,\infty]$ such that $0=\Psi(0)=\lim_{x\to 0}\Psi(x)$ and $\lim_{x\to\infty}\Psi(x)=\infty$.   The Orlicz space associated with $\Psi$ is defined as
$$L^\Psi:=L^\Psi(\Omega,\cF,\pr)=\big\{X\in L^0: \,\ex[\,\Psi(c|X|)\,]<\infty~
\mbox{ for some $c>0$}\big\}.
$$
It is a Banach space when endowed with the Luxemburg norm
 $$
    \|X\|_{\Psi} := \inf\left\{\lambda > 0\,:\,\ex [\,\Psi(|X|/\lambda)\,]\leq 1\right\}.
$$
The Orlicz heart is 
$$H^\Psi:=H^\Psi(\Omega,\cF,\pr)=\big\{X\in L^0\,:\,\ex [\,\Psi(c|X|)\,]<\infty~
\mbox{ for all $c>0$}\big\}.
$$
For example, if $\Psi(x)=x^p/p$ for some $p\in[1,\infty[,$ then $H^\Psi=L^\Psi=L^p := L^{p}\OFP$ is the usual $L^{p}-$space. In this case $ \|Y\|_{\Psi}=p^{-1/p}\|Y\|_p,$ where 
$\|\cdot\|_{p}$ stands for $L^{p}-$norm. If $\Psi$ takes the value $+\infty$, then $H^\Psi=\{0\}$ and $L^\Psi=L^\infty := 
L^{\infty}\OFP$ is defined to consist of all $\pr-$essentially bounded random variables. By Jensen inequality, we always have $H^{\Psi}\subseteq L^{1}.$ In the case of finite $\Psi,$ we see that $L^{\infty}$ is a linear subspace of $H^{\Psi},$ which is dense w.r.t. $\|\cdot\|_{\Psi}$ (see Theorem 2.1.14 in \cite{EdgarSucheston1992}).
\par
Let \(0<T<\infty\) and let $\OFFP$  be a filtered probability space, where 
$(\cF_{t})_{t\in [0,T]}$ is a right-continuous filtration with $\cF_{0}$ containing only the sets with probability $0$ or $1$ as well as all the null sets of $\cF$. Furthermore, consider a lower semicontinuous convex mapping $\Phi: [0,\infty[\rightarrow [0,\infty]$ satisfying 
$\Phi(x_{0}) < \infty$ for some $x_{0} > 0,$ $\inf_{x\geq 0}\Phi(x) = 0,$ and 
$\lim_{x\to\infty}\frac{\Phi(x)}{x} = \infty.$ Its Fenchel-Legendre transform
$$
\Phi^{*}:\R\rightarrow \R\cup\{\infty\},~y\mapsto\sup_{x\geq 0}~\big(xy - \Phi(x)\big)
$$
is a finite nondecreasing convex function whose restriction $\Phi^{*}\bigr |_{[0,\infty[}$ to 
$[0,\infty[$ is a finite Young function (cf. Lemma \ref{optimizedcertaintyequivalent} in Appendix \ref{AppendixAA}). We shall use  
$H^{\Phi^{*}}$ to denote the Orlicz heart w.r.t. $\Phi^{*}\bigr |_{[0,\infty[}.$ Then we can define a conditional convex  risk measure   $(\rho^{\Phi}_{t})_{t\in [0,T]}$ via 
$$
\rho^{\Phi}_{t}(X)=\esssup_{\Q\in\cQ_{\Phi,t}}\left(\ex_{\Q}[-X\,|\,\cF_{t}]- 
\ex\left[\left.\Phi\left(\frac{d\Q}{d\pr}\right)\right |\cF_{t}\right]\right)
$$
for all \(X\in H^{\Phi^{*}},\)  where $\cQ_{\Phi,t},$ denotes the set of all probability measures $\Q$ which are absolutely continuous w.r.t. $\pr$  such that $\Phi\left(\frac{d\Q}{d\pr}\right)$ is $\pr-$integrable and $\Q= \pr$ on \(\cF_{t}.\) Note that 
$\frac{d\Q}{d\pr} X$ is $\pr-$integrable for every $\Q\in \cQ_{\Phi,0}$ and any 
$X\in H^{\Phi^{*}}$ due to the Young's inequality. Consider now a right-continuous nonnegative stochastic process $(Y_{t})$  adapted to $(\mathcal{F}_t).$    Furthermore, let $\cT$ contain all finite stoping times $\tau\leq T$ w.r.t. $(\mathcal{F}_t).$ The main object of our study is the following optimal stopping problem
\begin{equation}
\label{stoppproblem}
\sup_{\tau\in\cT}\, \rho^{\Phi}_0(-Y_\tau).
\end{equation}
If we set $\Phi(x) = 0$ for $x\leq 1,$ and 
$\Phi(x) = \infty$ otherwise, we end up with the classical stopping problem
\begin{equation}
\label{riskneutral}
\sup_{\tau\in\cT}\ex[Y_{\tau}].
\end{equation}
It is well known that the optimal value of the problem \eqref{riskneutral} may be viewed as a risk neutral price of the American option with the discounted payoff $(Y_{t})_{t\in [0,T]}$ at time $t = 0.$ However, in face of incompleteness, it seems to be not appropriate to assume  the uniqueness of the risk neutral measure. Instead, the uncertainty about the stochastic process driving the payoff \(Y_t\) should be taken into account. Considering the optimal value of the problem (\ref{stoppproblem}) as an alternative pricing rule, model uncertainty risk is incorporated by taking the supremum over \(\cQ_{\Phi,t},\)  where the penalty function is used to assess the plausibility of possible models. The more plausible is the model, the lower is the value of the penalty function. 
\begin{example}
Let us illustrate our setup in the case of the so called Average Value at Risk risk measure. The Average Value at Risk risk measure at level $\alpha\in ]0,1]$ is defined as the following functional:
$$
AV@R_{\alpha}: \, X\mapsto -\frac{1}{\alpha}\,\int_{0}^{\alpha}
F^{\leftarrow}_{X}(\beta)\,d\beta,
$$
where \(X\in L^{1}\) and $F^{\leftarrow}_{X}$ denotes the left-continuous quantile function of the distribution function $F_{X}$ of $X$ {defined by $F^{\leftarrow}_{X}(\alpha) = \inf\{x\in\R\mid F_{X}(x)\geq\alpha\}$ for $\alpha\in ]0,1[$.} Note that $AV@R_{1}(X) = \ex[- X]$  for any $X\in L^{1}.$ Moreover, it is well known that 
$$
AV@R_{\alpha}(X) = \sup\limits_{\Q\in\cQ_{\Phi_{\alpha},0}}\ex_{\Q}[-X]\quad\mbox{for}\, X\in L^{1},
$$
where $\Phi_{\alpha}$ is the Young function defined by 
$\Phi_{\alpha}(x) = 0$ for $x\leq 1/\alpha,$ and $\Phi_{\alpha}(x) = \infty$ otherwise (cf. \cite[Theorem 4.52]{FoellmerSchied2010} and \cite{KainaRueschendorf2007}). Observe that the set $\cQ_{\Phi_{\alpha},0}$ consists of all probability measures on $\cF$ with 
$\frac{d\Q}{d\pr}\leq 1/\alpha\quad \pr-$a.s.. Hence the optimal stopping problem 
(\ref{stoppproblem}) reads as follows
\begin{equation}
\label{AV@Rstopping}
\sup_{\tau\in\cT} \, AV@R_{\alpha}(-Y_{\tau}) = \sup_{\tau\in\cT} \left\{\frac{1}{\alpha}\,\int_{1- \alpha}^{1}
F^{\leftarrow}_{Y_{\tau}}(\beta)\,d\beta \right\}.
\end{equation}
The family $\big(\rho^{\Phi_{\alpha}}_{t}\big)_{t\in [0,T]}$ of conditional convex risk measure associated with $\Phi_{\alpha}$ 
is also known as the conditional AV@R $(AV@R_{\alpha}(\cdot\,|\,\cF_{t}))_{t\in [0,T]}$ at level $\alpha$ 
(cf. \cite[Definition 11.8]{FoellmerSchied2010}). 
\end{example}
\begin{example}
\label{entropicriskmeasure}
Let us consider, for any $\gamma > 0,$ the continuous convex mapping $\Phi_{[\gamma]}: [0,\infty[\rightarrow\R$ defined by $\Phi_{[\gamma]}(x) = (x\ln(x) - x + 1)/\gamma$ for $x > 0$ and $\Phi_{[\gamma]}(0) = 1/\gamma.$ The Fenchel-Legendre transform of $\Phi_{[\gamma]}$ is given by
$\Phi_{[\gamma]}^{*}(y) = (\exp(\gamma y) - 1)/\gamma$ for $y\in\R.$ In view of Lemma \ref{optimizedcertaintyequivalent} (cf. Appendix \ref{AppendixAA}) the corresponding risk measure 
$\rho_{0}^{\Phi_{[\gamma]}}$ has the representation
\begin{equation}
\label{entropic}
\rho_{0}^{\Phi_{[\gamma]}}(X) = \inf_{x\in\R}\ex\left[\frac{\exp(\gamma x - \gamma X) - 1}{\gamma} - x\right] = 
\frac{\ln\big(\ex[\exp( - \gamma X)]\big)}{\gamma}
\end{equation}
for $X\in H^{\Phi_{[\gamma]}^{*}}.$ This is the well-known entropic risk measure. Optimal stopping with the entropic risk measures is easy to handle, since  it can be reduced to the standard optimal stopping problems via 
\begin{equation}
\label{entropicstop}
\sup_{\tau\in\cT}\rho_{0}^{\Phi_{[\gamma]}}(-Y_{\tau}) = 
\frac{1}{\gamma}\cdot\ln\big(\sup_{\tau\in\cT}\ex\Big[\exp(\gamma Y_{\tau})\big]\Big).
\end{equation}

\end{example}
\begin{example}
\label{polynomial}
Set $\Phi^{[p]} = x^{p}/p$ for  any $p\in ~]1,\infty[,$  then the set $\cQ_{\Phi^{[p]},0}$ contains all probability measures $\Q$ on $\cF$ with 
$\frac{d\Q}{d\pr}\in L^{p},$ and 
$$
\rho^{\Phi^{[p]}}(X) = \sup\limits_{\Q\in \cQ_{\Phi^{[p]},0}}\left(\ex_{\Q}[- X] - 
\frac{1}{p}~\ex\left[\left(\frac{d\Q}{d\pr}\right)^{p}\right]\right)\quad\mbox{for}~X\in L^{p/(p - 1)}.
$$  
\end{example}
\section{Main results}
\label{main_results}
Let $int(dom(\Phi))$ denote the topological interior of the effective domain of the mapping $\Phi: [0,\infty[\rightarrow [0,\infty].$ We shall assume $\Phi$ to be a lower semicontinuous convex function satisfying
\begin{equation}
\label{Annahmen Young function}
1\in int(dom(\Phi)),\quad\inf_{x\geq 0}\Phi(x) = 0,\quad\mbox{and},\quad\lim_{x\to\infty}\frac{\Phi(x)}{x} = \infty.  
\end{equation} 

\subsection{Primal representation}
The following theorem is our main result.
\begin{theorem}
\label{new_representation}
Let $(\Omega,\cF_{t},\left.\pr\right |_{\cF_{t}})$ be atomless with countably generated $\cF_{t}$ for every $t > 0.$ Furthermore, let  (\ref{Annahmen Young function}) be fulfilled, and  let $\sup_{t\in [0,T]}Y_{t}\in H^{\Phi^{*}},$ then 
\begin{eqnarray*}
\sup_{\tau\in\cT}\rho^{\Phi}_0 (-Y_\tau)
&=&
\sup\limits_{\tau\in\cT}\inf_{x\in\R}\ex[\Phi^{*}(x + Y_{\tau}) - x]\\ 
&=& 
\inf_{x\in\R}\sup\limits_{\tau\in\cT}\ex[\Phi^{*}(x + Y_{\tau}) - x] < \infty.
\end{eqnarray*}
\end{theorem}
\begin{remark}
\label{optimalstoppingoptimizedcertainty}
The functional
$$
\rho_{\Phi^{*}}: H^{\Phi^{*}}\rightarrow\R,\,X\mapsto \inf_{x\in\R} \ex[\Phi^{*}(x + X) - x]
$$
is known as the optimized certainty equivalent w.r.t. $\Phi^{*}$ (cf. 
\cite{Ben-TalTeboulle1987},\cite{Ben-TalTeboulle2007}). Thus the relationship 
\begin{eqnarray}
\label{main_relation}
\sup\limits_{\tau\in\cT}\inf_{x\in\R}\ex[\Phi^{*}(x + Y_{\tau}) - x] 
= 
\inf_{x\in\R}\sup\limits_{\tau\in\cT}\ex[\Phi^{*}(x + Y_{\tau}) - x] 
\end{eqnarray}
 may also be viewed as a representation result for optimal stopping with optimized certainty equivalents. 
\end{remark}
Let us illustrate Theorem \ref{new_representation}  for the case $\Phi = \Phi_{\alpha}$ with some $\alpha\in ]0,1].$ The Young function \(\Phi_{\alpha}\) satisfies the conditions of Theorem \ref{new_representation} if and only if $\alpha < 1.$ The Fenchel-Legendre transform $\Phi_{\alpha}^{*}$ of $\Phi$ is given by 
$\Phi_{\alpha}^{*}(x) = x^{+}/\alpha$ and it fullfills the inequality $\Phi_{\alpha}^{*}(x + y) - x\geq 
\Phi_{\alpha}^{*}(y)$ for $x,y\geq 0.$ Then, as an immediate consequence of Theorem \ref{new_representation}, we obtain the following primal representation for the optimal stopping problem 
(\ref{AV@Rstopping}).
\begin{corollary}
\label{representation AV@R}
Let $(\Omega,\cF_{t},\pr |_{\cF_{t}})$ be atomless with countably generated $\cF_{t}$ for every $t > 0.$ If 
$\sup_{t\in [0,T]}Y_{t}\in L^{1},$ then it holds for 
$\alpha\in ]0,1[$
\begin{eqnarray*}
\sup_{\tau\in\cT}AV@R_{\alpha}(-Y_{\tau})
&=&
\inf_{x\in\R}\sup\limits_{\tau\in\cT}\ex\left[\frac{1}{\alpha}\, (x + Y_{\tau})^{+} - x\right]\\
&=& 
\inf_{x\leq 0}\sup\limits_{\tau\in\cT}\ex\left[\frac{1}{\alpha}\, (x + Y_{\tau})^{+} - x\right] < \infty.
\end{eqnarray*}
\end{corollary}
Let us now consider the case $\Phi = \Phi^{[p]}$ for some $p\in ]1,\infty[.$ This mapping meets all  requirements of Theorem \ref{new_representation}, and $\Phi^{[p]^{*}}(x) = \Phi^{\left[p/(p-1)\right]}(x^{+}).$  Then by Theorem \ref{new_representation}, we have the following primal representation of the corresponding optimal stopping problem. 
\begin{corollary}
\label{primalrepresentation}
Let $(\Omega,\cF_{t},\pr |_{\cF_{t}})$ be atomless with countably generated $\cF_{t}$ for every $t > 0.$ If 
$\sup_{t\in [0,T]}Y_{t}\in L^{p/(p-1)}$ for some 
$p\in ]1,\infty[,$ then 
\begin{eqnarray*}
\sup_{\tau\in\cT}\rho^{\Phi^{[p]}}(-Y_{\tau})
=
\inf_{x\in\R}\sup\limits_{\tau\in\cT}\ex\left[\frac{(p-1)\, 
\big((x + Y_{\tau})^{+}\big)^{p/(p-1)}}{p} - x\right]
< \infty.
\end{eqnarray*}

\end{corollary}
\subsection{The existence of  solutions}
A natural question is whether we can find a real number \(x^*\) and a \((\mathcal{F}_{t})\)-stopping time \(\tau^*\) which solve \eqref{main_relation}. We may give a fairly general answer within the context of discrete time optimal stopping problems. In order to be more precise,  let $\cTT$ denote all stopping times from $\cT$ with values in $\TTT,$ where $\TTT$ is any finite subset of $[0,T]$ containing $\{0,T\}.$  Consider now the stopping problem 
\begin{equation}
\label{discretestopping}
\mbox{maximize}\quad 
\rho^{\Phi}_{0}(-Y_{\tau})
\,\mbox{over}\,\tau\in\cTT.
\end{equation}
Turning over to the filtration $(\cF^{\TTT})_{t\in [0,T]}$ defined by 
$\cF^{\TTT}_{t} := \cF_{[t]}$ with $[t] := {\max}\{s\in\TTT\mid s\leq t\},$ we see that 
$(Y^{\TTT})_{t\in [0,T]}$ with $Y^{\TTT}_{t} := Y_{[t]}$ describes some $(\cF^{\TTT}_{t})-$adapted process. Hence we can apply Theorem \ref{new_representation} to get 
\begin{eqnarray}
\label{discreteminimax} \nonumber
\sup_{\tau\in\cT_{\TTT}}\rho^{\Phi}_{0}(-Y_{\tau}) 
&=& 
\sup_{\tau\in\cT_{\TTT}}\inf_{x\in\R}\ex[\Phi^{*}(x + Y_{\tau}) - x]\\ 
&=&
\inf_{x\in\R}\sup_{\tau\in\cT_{\TTT}}\ex[\Phi^{*}(x + Y_{\tau}) - x]
\end{eqnarray}
In this section we want to find conditions which guarantee the existence of a saddle point for the  optimization problems
\begin{equation}
\label{primalproblem}
\mbox{maximize}\quad \inf_{x\in\R}\ex[\Phi^{*}(x + Y_{\tau}) - x]\,\mbox{over}\,\tau\in\cT_{\TTT}\,
\end{equation}
and
\begin{equation}
\label{dualproblem}
\mbox{minimize}\quad \sup_{\tau\in\cT_{\TTT}} \ex[\Phi^{*}(x + Y_{\tau}) - x]\,\mbox{over}\, x\in\R.
\end{equation}
To this end, we shall borrow some arguments from the theory of Lyapunoff's theorem for infinite-dimensional vector measures. A central concept in this context is the notion of \textit{thin subsets} of integrable mappings. So let us first recall it. For a fixed probability space $\tOFP,$ a subset $M\subseteq L^{1}\tOFP$ is called thin if for any $A\in\overline{\cF}$ with $\overline{\pr}(A) > 0,$ there is some 
nonzero $g\in L^{\infty}\tOFP$ vanishing outside $A$ and satisfying $\ex[g\cdot Z] = 0$ for every $Z\in M$ (cf. \cite{KingmanRobertson1968}, or \cite{Anantharaman2012}). Best known examples are finite subsets of $L^{1}\tOFP$ or 
finite-dimensional linear subspaces of $L^{1}\tOFP$ if 
$\tOFP$ is atomless (cf. \cite{KingmanRobertson1968}, or \cite{Anantharaman2012}).

\begin{proposition}
\label{saddle-point}
Let the assumptions of Theorem \ref{new_representation}
be fulfilled, and let 
$\TTT := \{t_{0},\dots,t_{r+1}\}$ with $t_{0} = 0< t_{1}<\dots< t_{r+1} = T.$
Moreover, let $\left\{\ex\left[\eins_{A}\cdot\Phi^{*}(x + Y_{s})~|~\cF_{t}\right]\mid x\in\R\right\}$ be a thin subset of $L^{1}(\Omega,\cF_{t},\pr|_{\cF_{t}})$ for $s,t\in\TTT$ with $t\leq s$ and $A\in\cF_{T}.$
Then there are $\tau^{*}\in\cT_{\TTT}$ and $x^{*}\in\R$ satisfying
\begin{eqnarray*}
\inf_{x\in\R}\ex[\Phi^{*}(x + Y_{\tau^{*}}) - x] 
&=& 
\sup\limits_{\tau\in\cT_{\TTT}}\inf_{x\in\R}\ex[\Phi^{*}(x + Y_{\tau}) - x]\\ 
&=& 
\inf_{x\in\R}\sup_{\tau\in\cT_{\TTT}}\ex[\Phi^{*}(x + Y_{\tau}) - x]\\
&=& 
\sup_{\tau\in\cT_{\TTT}}\ex[\Phi^{*}(x^{*} + Y_{\tau}) - x^{*}].
\end{eqnarray*}
In particular, it holds
\begin{eqnarray*}
\ex[\Phi^{*}(x^{*} + Y_{\tau}) - x^{*}] 
\leq 
\ex[\Phi^{*}(x^{*} + Y_{\tau^{*}}) - x^{*}] 
\leq 
\ex[\Phi^{*}(x + Y_{\tau^{*}}) - x]
\end{eqnarray*}
for any $x\in\R$ and $\tau\in\cT_{\TTT}.$
\end{proposition}

The proof of Proposition \ref{saddle-point} can be found in Section \ref{proof of saddle-point}.

\begin{example} 
Let the mapping $\Phi_{e}^{*}:\R\rightarrow\R$ be defined by $\Phi_e^{*}(y) := \sum_{k = 1}^{n}\alpha_{k}(\exp(\beta_{k}y) - 1)$ for some $\alpha_{1},\dots,\alpha_{n}, \beta_{1},\dots,\beta_{n} > 0.$ Obviously, \(\Phi_e^{*}\) is convex, nondecreasing, and satisfies $\lim_{y\to\infty}(\Phi_e^{*}(y) - y) = \infty$ as well as $\Phi_e^{*}(0) = 0.$  Hence $\Phi_e(x) := \sup_{y\in\R}(xy - \Phi_e^{*}(y))$ defines a lower semicontinuous convex function  which satisfies \eqref{Annahmen Young function}, and whose Fenchel-Legendre transform coincides with $\Phi_e^{*},$ since $\Phi_e^{*}$ is continuous.
Moreover,  for any $s,t\in\TTT$ such that $t\leq s,$ and $A\in\cF_{T},$ the set $\left\{\ex\left[\eins_{A}\cdot\Phi_e^{*}(x + Y_{s})~|~\cF_{t}\right]\mid x\in\R\right\}$ is contained in the finite-dimensional linear subspace of $L^{1}(\Omega,\cF_{t},\pr|_{\cF_{t}})$ spanned by the sequence of r. v. 
$$
\left\{\ex\left[\eins_{A}\cdot\exp(\beta_k Y_s)~|~\cF_{t}\right] ~\big|~ \, k=0,\dots,n\right\},
$$
where by definition \({{\beta}_0} := 0.\)
 As a result, 
$\left\{\ex\left[\eins_{A}\cdot\Phi_e^{*}(x + Y_{s})~|~\cF_{t}\right]\mid x\in\R\right\}$ is a thin subset of $L^{1}(\Omega,\cF_{t},\pr|_{\cF_{t}})$ in the case of atomless $(\Omega,\cF_{t},\pr|_{\cF_{t}})$ (cf. e.g. \cite[Proposition 2.6]{Anantharaman2012}). \end{example}
\subsection{Additive dual representation}
\label{add_dual}
 In this section we generalize the celebrated additive dual representation for optimal stopping problems (see \cite{Rogers2002}) to the case of optimal stopping under uncertainty. The result in \cite{Rogers2002} is formulated in terms of martingales $M$ with $M_{0} = 0$ satisfying  $\sup_{t\in [0, T]}|M_{t}|\in L^{1}.$ The set of all such adapted martingales will be denoted by $\cM_{0}.$ 
\begin{theorem}
\label{dualrepresentation}
Let 
$V_{t} := \esssup_{\tau\in\cT, \tau\geq t}\ex\left[Z_{\tau}\,|\,\cF_{t}\right]$ be the Snell enve\-lope of an integrable right-continuous stochastic process $(Z_{t})_{t\in [0,T]}$ adapted to $\OFFP.$ 
If  $\sup_{t\in [0, T]}|Z_{t}|\in L^{p}$ for some $p > 1,$ then 
$$
V_{0} = \sup_{\tau\in \cT}\ex[Z_{\tau}] = \inf_{M\in \cM_{0}}\ex\left[\sup_{t\in [0,T]}(Z_{t} - M_{t})\right],
$$
where the infimum is attained for $M = M^{*}$ with $M^{*}$ being the martingale part of the Doob-Meyer decomposition of $(V_{t})_{t\in [0,T]}.$ Even more it holds 
$$
\sup_{\tau\in\cT}\ex[Z_{\tau}] = \sup_{t\in [0,T]}({Z}_{t} - M^{*}_{t})\quad\pr-\mbox{a.s.}.
$$
\end{theorem} 
\begin{remark}
\label{relaxed}
By inspection of the proof of Theorem 2.1 in \cite{Rogers2002}, one can see that the assumption  $\sup_{t\in [0, T]}\ex[Z_{t}]\in L^{p}$ for some $p > 1$ is only used to guarantee the existence of the Doob-Meyer decomposition of the Snell envelope $(V_{t})_{t\in [0,T]}.$ Therefore this assumption may be relaxed, if  we consider discrete time optimal stopping problems on the set $\TTT$ for some finite $\TTT\subseteq [0,T]$ containing $\{0,T\}.$ In this case, the Doob-Meyer decomposition always exists if $(Z_{t})_{t\in\TTT}$ is integrable, and   Theorem \ref{dualrepresentation} holds  with \(\cT\) replaced by \(\cT_\TTT\) and \([0,T]\) replaced by \(\TTT\) (see also 
\cite[Theorem 5.5]{KraetschmerSchoenmakers2010}).
\end{remark}
Theorem~\ref{new_representation} allows us to extend the additive dual representation to the case of stopping problems 
\eqref{stoppproblem}. We shall use the following notation. For a fixed $\Phi$ and $x\in\R$ we shall denote by $V^{\Phi,x} = (V^{\Phi,x}_{t})_{t\in [0,T]}$ the Snell-envelope  w.r.t. to $\big(\Phi^{*}(x + Y_{t}) - x\big)_{t\in [0,T]}$ defined via
\[
V^{\Phi,x}_{t} := \esssup_{\tau\in\cT, \tau\geq t}\ex\left[(\Phi^{*}(x +Y_{\tau}) - x)\,|\,\cF_{t}\right].
\] 
The application of Theorem \ref{new_representation} together with Theorem \ref{dualrepresentation} provides us with the following additive dual representation of the stopping problem \eqref{stoppproblem}. 
\begin{theorem}
\label{dualrepresentation_utility}
Under assumptions on $\Phi$ and $(\cF_{t})$ of Theorem~\ref{new_representation} and under the condition  $\sup_{t\in [0,T]}|\Phi^{*}(x + Y_{t})|\in L^{p}$ for some $p > 1$ and any $x\in \R,$ the following dual representation holds 
\begin{eqnarray*}
\sup_{\tau\in\cT}\rho^{\Phi}_0 (-Y_\tau) 
\hspace*{-0.25cm}&=&
\inf_{x\in \R}\inf_{M\in \cM_{0}}\ex\big[\sup_{t\in [0,T]}\big(\Phi^{*}(x + Y_{t}) - x - M_{t}\big)\big]\\
&=&
\inf_{x\in\R}\ex\big[\sup_{t\in [0,T]}\big(\Phi^{*}(x + Y_{t}) - x - M^{*,\Phi,x}_{t}\big)\big]\\
&=& 
\essinf_{x\in\R}\sup_{t\in [0,T]}\big(\Phi^{*}(x + Y_{t}) - x - M^{*,\Phi,x}_{t}\big)
\quad \pr-\mbox{a.s.}.
\end{eqnarray*}
Here 
$M^{*,\Phi,x}$ stands for the martingale part of the Doob-Meyer decomposition of the Snell-envelope $V^{\Phi,x}.$
%
\end{theorem}
\begin{remark}
\label{relaxed2}
Under the assumptions of Theorem~\ref{new_representation}, we have that $\sup_{t\in [0,T]}Y_{t}\in H^{\Phi^{*}}.$ Furthermore, $\Phi^{*}$ is convex and nondecreasing with $\Phi^{*}(0) = 0$ (see Lemma \ref{optimizedcertaintyequivalent} in Appendix \ref{AppendixAA}) so that  for any $y < 0$
$$
|\Phi^{*}(y)| = \int_{y}^{0}\Phi^{*'}(z)\,dz\leq \Phi^{*'}(0) |y|\leq 
\int_{0}^{|y|}\Phi^{*'}(z)\,dz = \Phi^{*}(|y|),
$$
where $\Phi^{*'}$ denotes the right-sided derivative of $\Phi^{*}.$ Using the monotonicity of $\Phi^{*}$ again, we  conclude that
$$
|\Phi^{*}(x + Y_{t})|\leq \Phi^{*}(|x| + Y_{t})\leq \Phi^{*}(|x| + \sup_{t\in [0,T]}Y_{t})\in L^{1}
$$
for all $x\in\R$ and $t\in [0,T].$ Hence the application of Theorem~\ref{dualrepresentation_utility} to \eqref{discretestopping} is already possible under the assumptions of Theorem \ref{new_representation}.
\end{remark}
The dual representation for the optimal stopping problem under Average Value at Risk reads as follows.
\begin{corollary}
\label{cor_dual_avar}
Let the assumptions on $\Phi$ and $(\cF_{t})$ be as in Theorem~\ref{new_representation}. If  
$\sup_{t\in [0,T]}Y_{t}\in L^{p}$ for some $p > 1,$  then it holds 
$\pr$-a.s.
\begin{eqnarray}
\nonumber
\sup_{\tau\in\cT}AV@R_{\alpha}(-Y_{\tau})
&=&
\nonumber
\inf_{x\in\R}\inf_{M\in \cM_{0}}\ex\left[\sup_{t\in [0,T]}\left(\frac{1}{\alpha}\, (x + Y_{t})^{+} - x - M_{t}\right)\right]\\
&=&
\nonumber
\inf_{x\leq 0}\ex\left[\sup_{t\in [0,T]}\left(\frac{1}{\alpha}\, (x + Y_{t})^{+} - x - M^{*,\alpha,x}_{t}\right)\right]\\
&=&
\label{dual_avr}
\essinf_{x\leq 0}\sup_{t\in [0,T]}\left(\frac{1}{\alpha}\, (x + Y_{t})^{+} - x - M^{*,\alpha,x}_{t}\right)\quad \pr-\mbox{a.s.}.
\end{eqnarray}
Here 
$M^{*,\alpha,x}$ denotes the martingale part of the Doob-Meyer decomposition of 
the Snell-envelope $V^{\Phi_{\alpha},x}.$
%
\end{corollary}
\begin{remark}
\label{relax3}
Let us consider a discrete time optimal stopping problem $\sup_{\tau\in\cT_{\TTT}}AV@R_{\alpha}(-Y_{\tau})$ for some finite $\TTT\subseteq [0,T]$ with $\{0,T\}\in\TTT.$ In view of Remark \ref{relaxed2}, the assumptions of Theorem \ref{new_representation} are already sufficient to obtain the dual representation \eqref{dual_avr} with $\cT$ replaced by $\cT_{\TTT}$ and \([0,T]\) replaced by \(\TTT.\) 
\end{remark}
\section{Discussion}
\label{discussion}
In \cite{KraetschmerSchoenmakers2010}   the optimal stopping problems of the type  
\begin{equation}
\label{DMU}
\sup_{\tau\in \cT}\,\cU_{0}(Y_{\tau})
\end{equation}
were studied, where  for any \(t\geq 0,\) the functional $\cU_{t}$  maps a linear subspace 
$\cX$ of the space $L^{0}$ into $\cX\cap L^{0}(\Omega,\cF_{t}, \pr |_{\cF_{t}})$ and satisfies $\cU_{t}(X)\leq\cU_{t}(Y)$ for $X\leq Y\quad\pr-$a.s.. In fact there is a one-to-one correspondence between conditional convex risk measures $(\rho_{t})_{t\in [0,T]}$  and dynamic utility functionals $\cU := (\cU_{t})_{t\in [0,T]}$  satisfying the following two properties:
\begin{itemize}
\item {\bf conditional translation invariance:}\\ 
$\cU_{t}(X + Y) = \cU_{t}(X) + Y$ for 
$Y\in\cX\cap L^{0}(\Omega,\cF_{t},\pr|_{\cF_{t}})$ and $X\in\cX,$
\item {\bf conditional concavity:}\\
$\cU_{t}(\lambda X + (1 - \lambda)Y) \geq \lambda\,\cU_{t}(X) + (1 - \lambda)\cU_{t}(Y)$ for 
$X, Y\in\cX$ and $\lambda\in\cX\cap L^{0}(\Omega,\cF_{t},\pr|_{\cF_{t}})$ with 
$0\leq \lambda\leq 1.$ 
\end{itemize}
More precisely, any conditionally translation invariant and conditionally concave dynamic utility functional $(\cU_{t})_{t\in [0,T]}$ defines a family $(\rho^{\cU}_{t})_{t\in [0,T]}$ of conditional convex risk measures via $\rho^{\cU}_{t}(X) = -\cU_{t}(X)$ and vice versa. The results of \cite{KraetschmerSchoenmakers2010} essentially rely  on the following additional assumptions  
\begin{itemize}
\item {\bf regularity:}\\ 
$\cU_{t}(\eins_{A}X) = \eins_{A}\cdot\cU_{t}(X)$ for $A\in\cF_{t}$ and $X\in\cX,$
\item {\bf recursiveness:}\\
$\cU_{s}\circ\cU_{t} = \cU_{s}$ for $s\leq t.$
\end{itemize}
{Recursiveness is often also referred to as time consistency.} Obviously, the dynamic utility functional $(\cU_{t}^{\Phi_{\alpha}})_{t\in [0,T]},$ defined by $\cU_{t}^{\phi_{\alpha}}(X) := AV@R_{\alpha}(-X|\cF_{t}),$ satisfies the regularity and the conditional translation invariance, but it fails to be recursive (cf. \cite[Example, 11.13]{FoellmerSchied2010}). Even worse, according to Theorem 1.10 in \cite{KupperSchachermayer2009} for any $\alpha < 1,$ there is in general no regular conditionally translation invariant and recursive dynamic utility functional $\cU$ such that $\cU_{0} = \cU^{\Phi_{\alpha}}_{0}.$ This means that  we can not in general reduce the stopping problem (\ref{AV@Rstopping}) to 
the stopping problem (\ref{DMU}) with a regular, conditionally translation invariant and recursive dynamic utility functional \(\cU\). Note that this conclusion can be drawn from Theorem 1.10 of \cite{KupperSchachermayer2009}, because $AV@R_{\alpha}$ is law-invariant, i.e., 
$AV@R_{\alpha}(X) = AV@R_{\alpha}(Y)$ for identically distributed $X$ and $Y$, and satisfies the properties $AV@R_{\alpha}(0) = 0$ as well as 
$AV@R_{\alpha}(-\varepsilon\eins_{A}) > 0$ for any $\varepsilon>0$ and $A\in\cF$ with 
$\pr(A) > 0.$  
\par
The stopping problem (\ref{AV@Rstopping}) may also be viewed as a special case of the following stopping problem:
\begin{eqnarray}
\label{osf_distortion}
\sup_{\tau \in \mathcal{T}}\int _0^\infty w(\pr(Y_\tau > x))\, dx,
\end{eqnarray}
where \(w:[0,1]\mapsto[0,1]\) is a so-called distortion function, i.e., $w$ is nondecreasing and satisfies $w(0) = 0,$ $w(1) = 1.$ Indeed, if for  
$\alpha\in ]0,1[$ the distortion function $w_{\alpha}$ is defined by  
$w_{\alpha}(u) :=  \frac{u}{{\alpha}}\wedge 1,$ then the stopping problems (\ref{AV@Rstopping}) and (\ref{osf_distortion}) coincide.  Recalling Theorem 1.10 of \cite{KupperSchachermayer2009} again, we see that  the stopping problem 
(\ref{osf_distortion}) is not in general  representable in the form (\ref{DMU}) with some regular, conditionally translation invariant and recursive dynamic utility functional. 
The stopping problem (\ref{osf_distortion}) was recently considered by \cite{XuZhou2013}. However, the analysis in \cite{XuZhou2013} relies on some additional assumptions. First of all, the authors allow for all finite stopping times w.r.t. to some filtered probability space $(\Omega,\cF,(\cF_{t})_{t\geq 0},\pr)$ instead of restricting to those which are bounded by a fixed number. Secondly, they assume a special structure for the process $(Y_{t})_{t\geq 0},$ namely  it is supposed that $Y_{t} = u(S_{t})$ for  an absolutely continuous nonnegative function  $u$ on $[0,\infty[$ and for  a one-dimensional geometric Brownian motion $(S_{t})_{t\geq 0}$.  Thirdly, the authors focus on strictly increasing absolutely continuous distortion functions $w$ so that their analysis does not cover the case of Average Value at Risk. More precisely, 
in \cite{XuZhou2013} the optimal stopping problems of the form 
\begin{eqnarray}
\label{os_distortion}
\sup_{\tau \in \mathcal{T}^{\infty}} \, D_{w}(u(S_\tau))=\sup_{\tau \in \mathcal{T}^{\infty}}\int _0^\infty w(\pr(u(S_\tau)>x))\, dx,
\end{eqnarray}
are studied, where $\mathcal{T}^{\infty}$ denotes the set of all finite stopping times. A crucial step  in the authors' argumentation is the reformulation of the optimal stopping problem \eqref{os_distortion} as 
\begin{eqnarray*}
\sup_{\tau \in \mathcal{T}^{\infty}} \, D_{w}(u(S_\tau)) 
&=& 
\sup_{F\in {\cal D}} \int_{0}^{\infty}w(1 - F(x))u'(x)\,dx \\
&=& 
\sup_{F\in {\cal D}} \int_{0}^{1} u(F^{\leftarrow}(u)) w'(1 - u)\, du,
\end{eqnarray*}
where $u'$ and $w'$ are derivatives of $u$ and $w,$ respectively, and ${\cal D}$ denotes the set of all distribution functions $F$ with a nonnegative support such that 
$\int_{0}^{\infty}(1 - F(x))\,dx\leq S_{0}.$ The main idea of the approach in \cite{XuZhou2013} is that any such distribution function may be described as the distribution function of $S_{\tau}$ for some finite stopping time $\tau\in\cT^{\infty}$ and this makes the application of the Skorokhod embedding technique possible. Hence, the results essentially rely on the special structure of the stochastic process $(Y_{t})_{t\geq 0}$ and   seem to be not extendable to stochastic processes of the form $Y_{t} = U(X_{t}),$ where $(X_{t})_{t\geq 0}$ is a multivariate Markov process. Moreover,  it remains unclear whether the analysis of  \cite{XuZhou2013} can be carried over to the case of bounded stopping times, as  the Skorokhod embedding can not be applied to  the general sets of stopping times $\cT$ (see e.g. \cite{AnkirchnerStrack2011}). 

\section{Numerical example}
\label{MC}

In this section we illustrate how our results can be used to price Bermudan-type  options in uncertain environment. 
Specifically, we consider the model with $d$ identically distributed assets, where each underlying has dividend yield $\delta $.
The  dynamic of assets is given by
\begin{equation}
\label{Xeq}
\frac{dX_{t}^{i}}{X_{t}^{i}}=(r-\delta )dt+\sigma dW_{t}^{i},\quad i=1,\ldots,d,
\end{equation}%
where $W_{t}^{i},\,i=1,\ldots,d$, are independent one-dimensional Brownian
motions and $r,\delta ,\sigma $ are constants. At any time $t\in
\{t_{0},\ldots,t_{J}\}$ the holder of the option may exercise it and
receive the payoff
\begin{equation*}
Y_t=G(X_{t})=e^{-rt}(\max (X_{t}^{1},...,X_{t}^{d})-K)^{+}.
\end{equation*}%
If we are uncertain about our modelling assumption  and  if the Average Value at Risk is used to measure the risk related to this uncertainty, then the risk-adjusted
price of the option is given by
\begin{eqnarray}
\nonumber
\sup_{\tau\in\cT[t_0,\ldots, t_J]}AV@R_{\alpha}(-Y_{\tau})&=& {\sup_{\tau\in\cT[t_0,\ldots, t_J]}\sup\limits_{\Q\in\cQ_{\Phi_{\alpha},0}}\ex_{\Q}[-Y_{\tau}]}
\\
\label{opt_stop_ex}
&=& \inf_{x\leq 0}\sup\limits_{\tau\in\cT[t_0,\ldots, t_J]}\ex\left[\frac{1}{\alpha}\, (x + Y_{\tau})^{+} - x\right],
\end{eqnarray}
{
where $\cQ_{\Phi_{\alpha},t}$ consists of all probability measures \(Q\) on $\cF$ with 
\begin{eqnarray}
\label{class_avr_ex}
\left.\frac{d\Q}{d\pr}\right|_{\mathcal{F}_t}\leq 1/\alpha, \quad \pr|_{\cF_t}-\text{ a.s.}.
\end{eqnarray}
If we restrict our attention to the class of  generalised Black Scholes models of the type
\begin{eqnarray*}
dX^i_t=X^i_t\,(\alpha^i_t\,dt+\sigma^i_t\, dW^i_t),\quad i=1,\ldots, d
\end{eqnarray*}
with adapted processes \((\alpha^i_{t}),\) \((\sigma^i_{t})\) and independent Brownian motions \(W^1_{t},\ldots, W^d_{t}, \)  then
\begin{eqnarray*}
\left.\frac{d\Q}{d\pr}\right|_{\mathcal{F}_t}=\exp\left(-\sum_{i=1}^d \int_0^t \theta^i_s\, dW^i_s-\frac{1}{2}\sum_{i=1}^d\int_{0}^t(\theta^i_s)^2\,ds\right)
\end{eqnarray*}
with \(\theta^i_t=(\alpha^i_t-r+\delta)/\sigma^i_t\) and the condition \eqref{class_avr_ex} transforms to
\begin{eqnarray*}
\exp\left(-\sum_{i=1}^d \int_0^t \theta^i_s\, dW^i_s-\frac{1}{2}\sum_{i=1}^d\int_{0}^t(\theta^i_s)^2\,ds\right)\leq 1/\alpha, \quad \pr|_{\cF_t}-\text{ a.s.}.
\end{eqnarray*}
}
Due to Corollary~\ref{representation AV@R},  one can use the standard methods based on dynamic programming principle to 
solve \eqref{opt_stop_ex} and \(\cT[t_0,\ldots, t_J]\) stands for a set of stopping times with values in \(\{t_0,\ldots, t_J\}.\) 
Indeed, for any fixed \(x,\) the optimal value of the stopping problem 
\begin{eqnarray*}
V=\sup\limits_{\tau\in\cT[t_0,\ldots, t_J]}\ex\left[\frac{1}{\alpha}\, (x + Y_{\tau})^{+} - x\right]
\end{eqnarray*}
can be, for example, numerically approximated via the well known regression methods like Longstaff-Schwartz method. In this way one can get a (suboptimal) stopping rule 
\[
\widehat\tau_x:=\inf\Bigl\{0\leq j\leq J: (x + Y_{t_j})^{+}/\alpha - x\geq \widehat C_j(X_{t_j},x)\Bigr\},
\] 
where \(\widehat C_1,\ldots,\widehat C_J\) are continuation values estimates. Then
\begin{eqnarray}
\label{low_bound}
V^l_N:=\inf_{x\leq 0}\left\{\frac{1}{N}\sum_{n=1}^N \Bigl(x + Y^{(n)}_{t_{\widehat\tau^{(n)}_x}}\Bigr)^{+}/\alpha - x\right\}
\end{eqnarray}
is a low-biased estimate for \(V\).  Note that the infimum in \eqref{low_bound} can be easily computed using a simple search algorithm. An upper-biased estimate can be constructed using the well known Andersen-Broadie dual approach (see \cite{AndersenBroadie2004}). For any fixed \(x\leq 0\) this approach would give us a discrete time martingale \((M_j^x)_{j=0,\ldots,J}\) which in turn can be used to build an upper-biased estimate via the representation (\ref{dual_avr}):
\begin{eqnarray}
\label{upper_bound}
V^u_N:=\inf_{x\leq 0}\left\{\sum_{n=1}^N\left[\sup_{j=0,\ldots,J}\left(\frac{1}{\alpha}\, \Bigl(x + Y^{(n)}_{t_j}\Bigr)^{+} - x - M^{x,(n)}_{j}\right)\right]\right\}.
\end{eqnarray}
Note that \eqref{upper_bound} remains upper biased even if we replace the infimum of the objective function in \eqref{upper_bound} by its value at a fixed point \(x.\)
In Table~\ref{max_call_2d} we present the bounds \(V^l_N\) and \(V^u_N\) together with their standard deviations for different values of \(\alpha.\) As to implementation details, we used \(12\) basis functions for regression (see \cite{AndersenBroadie2004}) and \(10^4\) training paths to compute 
\(\widehat C_1,\ldots,\widehat C_J.\) In the dual approach of Andersen and Broadie, \(10^3\) inner simulations were done to approximate \(M^x.\) In both cases we simulated \(N=10^4\) testing paths  to compute the final estimates.
\par
{
For comparison let us consider a problem of pricing the above Bermudan option under entropic risk measure \eqref{entropic}. Due to \eqref{entropicstop}, we need to solve the optimal stopping problem 
\begin{equation*}
V^\gamma=\sup_{\tau\in\cT[t_0,\ldots, t_J]}\ex\big[\exp(\gamma Y_{\tau})\big].
\end{equation*}
The latter problem can be solved via the standard dynamic programming combined with regression as described
above. In Table~\ref{max_call_2d_entrop} the upper and lower MC bounds for \(\log(V)/\gamma\) are presented for different values of the parameter \(\gamma.\) Unfortunately for larger values of \(\gamma,\) the corresponding MC estimates become unstable due to the presence of exponent in \eqref{entropicstop}.
In Figure~\ref{fig:bounds} the lower bounds for AV@R and the entropic risk measure are shown graphically. 
As can be seen the quality of upper and lower bounds are quite similar. However due to above mentioned instability, AV@R should be preferred under higher uncertainty. 
}

\begin{figure}[h]
\centering
\includegraphics[width=0.9\textwidth]{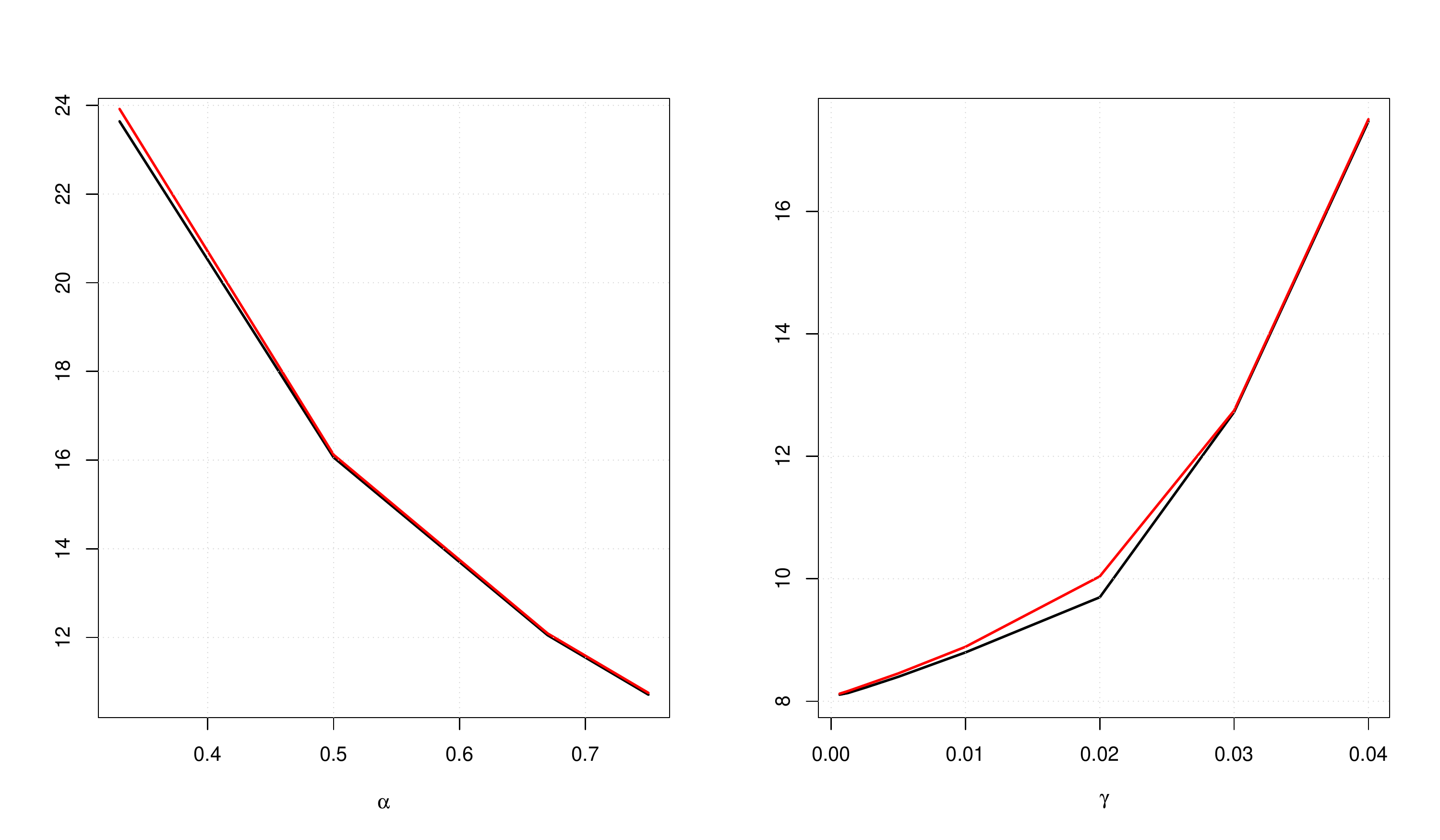}\caption{Lower and upper bounds for Bermudan option prices under AV@R (left) and entropic risk (right) measures.}%
\label{fig:bounds}%
\end{figure}
\begin{table}
\label{max_call_2d}
\caption{Bounds (with standard deviations) for \(2\)-dimensional Bermudan max-call  with parameters $K=100,\,r=0.05$, $\sigma=0.2,$ $\delta=0.1$ under AV@R at level \(\alpha\)}
\begin{tabular}{|c|c|c|}
\hline
\(\alpha\)  & Lower bound \(V^l_N\) & Upper bound \(V^u_N\)\\
\hline
\hline
0.33 & 23.64(0.026) & 23.92(0.108) \\
0.50 & 16.06(0.019)  & 16.12(0.045) \\
0.67 & 12.05(0.014)  & 12.09(0.034) \\
0.75 & 10.71(0.013)   & 10.75(0.030) \\
\hline
\end{tabular}
\end{table}

\begin{table}[h]
\label{max_call_2d_entrop}
\caption{Bounds (with standard deviations) for \(2\)-dimensional Bermudan max-call  with parameters $K=100,\,r=0.05$, $\sigma=0.2,$ $\delta=0.1$ under entropic risk measure with parameter \(\gamma\) }  
\begin{center}
 \begin{tabular}{|c|c|c|}
\hline 
$\gamma$ & Lower bound & Upper bound \\
\hline
\hline
0.0025& 8.218979 (0.011) &8.262082 (0.029) \\ 
0.005&   8.399141 (0.015) & 8.454748 (0.032) \\
0.01& 8.797425 (0.017) &  8.888961 (0.041)\\ 
0.02& 9.698094 (0.020) & 10.03958 (0.058)\\
0.03& 12.72327 (0.020) &12.74784 (0.072)\\
0.04& 17.47090 (0.022) & 17.50481 (0.095)\\ 
\hline
\end{tabular}
 \end{center}
\end{table}

\section{The optimal stopping problem with randomized stopping times}
\label{optimalrandomizedstoppingtimes}

In order to {prove} Theorem \ref{new_representation} we shall proceed as follows. 
First, by Lemma \ref{optimizedcertaintyequivalent} (cf. Appendix \ref{AppendixAA}), we obtain immediately 
\begin{equation}
\label{Hilfsstoppproblem}
\sup\limits_{\tau\in\cT}\sup_{\Q\in\cQ_{\Phi,0}}\left(\ex_{\Q}[Y_{\tau}]- \ex\left[\Phi\left(\frac{d\pr}{d\Q}\right)\right]\right)
= 
\sup\limits_{\tau\in\cT}\inf_{x\in\R}\ex[\Phi^{*}(x + Y_{\tau}) - x].
\end{equation}
The proof of Theorem \ref{new_representation} would be completed, if we can show that
\begin{equation}
\label{minimaxrelationship}
\sup\limits_{\tau\in\cT}\inf_{x\in\R}\ex[\Phi^{*}(x + Y_{\tau}) - x] 
= 
\inf_{x\in\R}\sup\limits_{\tau\in\cT}\ex[\Phi^{*}(x + Y_{\tau}) - x].
\end{equation}
Using Fubini's theorem, we obtain for any $\tau\in\cT$ and every $x\in\R$
$$
\ex[\Phi^{*}((x + Y_{\tau})^{+}) - x] = \int_{x^{-}}^{\infty}\Phi^{*'}(x + z)[1 - F_{Y_{\tau}}(z)]~dz + \Phi^{*}(x^{+}) - x,
$$
where $F_{Y_{\tau}}$ stands for the distribution function of $Y_{\tau}$ and $\Phi^{*'}$ denotes the right-sided derivative of the convex function $\Phi^{*}.$ In the same way we may also find 
$$
\ex[\Phi^{*}(-(x + Y_{\tau})^{-})] = -\int_{0}^{x^{-}}\Phi^{*'}(x + z)F_{Y_{\tau}}(z)~dz.
$$
Hence the property $\Phi^{*}(x) = \Phi^{*}(x^{+}) + \Phi^{*}(- x^{-})$ for $x\in\R$ yields
\begin{equation}
\label{DarstellungZielfunktion} 
\ex[\Phi^{*}(x + Y_{\tau}) - x]
=
\int_{0}^{\infty}\Phi^{*'}(x + z)[1 - F_{Y_{\tau}}(z)]~dz + \Phi^{*}(x) - x
\end{equation}
for $\tau\in\cT$ and $x\in\R.$ Since the set $\mathbb{F} := \{F_{Y_{\tau}}\mid \tau\in\cT\}$ of distribution functions 
$F_{Y_{\tau}}$ of $Y_{\tau}$ is not, in general, a convex subset of the set of distribution functions on $\R,$ we can not apply the known minimax results. The idea is to first establish (\ref{minimaxrelationship}) for the larger class of randomized stopping times, and then to show  that the optimal value coincides with the optimal value $\sup_{\tau\in\cT}\inf_{x\in\R}\ex[\Phi^{*}(x + Y_{\tau}) - x].$ 

Let us recall the notion of randomized stopping times. By definition 
(see e.g. \cite{EdgarMilletSucheston1981}), a randomized stopping time w.r.t. $\OFFP$ is 
a mapping $\tau^{r}:\Omega\times [0,1]\rightarrow [0,\infty]$ which is nondecreasing and left-continuous in the second component such that $\tau^{r}(\cdot,u)$ is a stopping time w.r.t. $(\cF_{t})_{t\in [0,T]}$ for any $u\in [0,1].$ Notice that any randomized stopping time $\taur$ is also an ordinary stopping time w.r.t. the enlarged filtered probability space 
$\big(\Omega\times [0,1],\cF\otimes \cB([0,1]),
\big(\cF_{t}\otimes \cB([0,1])\big)_{t\in [0,T]},\pr\otimes \pr^{U}\big).$ Here 
$\pr^{U}$ denotes the uniform distribution on $[0,1],$ defined on $\cB([0,1]),$ the usual Borel $\sigma-$algebra on $[0,1].$ We shall call a randomized stopping time $\taur$ to be degenerated if $\taur(\omega,\cdot)$ is constant for every $\omega\in\Omega.$ There is an obvious one-to-one correspondence between stopping times and degenerated randomized stopping times.
\par
Consider the stochastic process $(Y_{t}^{r})_{t\geq 0},$ defined by 
$$
Y_{t}^{r} :\Omega\times [0,1]\rightarrow\R,\, (\omega,u)\mapsto Y_{t}(\omega).
$$
which is adapted w.r.t. the enlarged filtered probability space. 
Denoting by $\cTr$ the set of all randomized stopping times $\tau^{r}\leq T,$ we shall study the following new stopping problem
\begin{equation}
\label{randomstop}
\mbox{maximize } \inf_{x\in\R}\ex[\Phi^{*}(x + Y^{r}_{\taur}) - x]~\mbox{over}~\tau^r\in\cTr.
\end{equation}
Obviously, $\inf_{x\in\R}\ex[\Phi^{*}(x + Y_{\tau}) - x] = \inf_{x\in\R}\ex[\Phi^{*}(x + Y^{r}_{\taur}) - x]$ is valid for every stopping time $\tau\in\cT,$ where $\taur\in\cTr$ is the corresponding degenerated randomized stopping time such that \(\taur(\omega,u)=\tau(\omega),\) \(u\in [0,1].\) Thus, in general the optimal value of the stopping problem (\ref{randomstop}) is at least as large as the one of the original stopping problem 
(\ref{stoppproblem}) due to (\ref{Hilfsstoppproblem}). One reason to consider the new stopping problem \eqref{randomstop} is that it  has a solution under fairly general conditions.
\begin{proposition}
\label{solution}
{Let $(Y_{t})_{t\in [0,T]}$ be quasi-left-continuous, defined to mean $Y_{\tau_{n}}\to Y_{\tau}$ $\pr-$a.s. whenever $(\tau_{n})_{n\in\N}$ is a sequence in $\cT$ satisfying $\tau_{n}\nearrow\tau$ for some $\tau\in\cT$.} If $\cF_{T}$ is countably generated, then there exists a randomized stopping time $\taur_*\in\cTr$ such that
$$
\inf_{x\in\R}\ex\left[\Phi^{*}(x + Y^{r}_{\taur_*}) - x\right] 
= 
\sup_{\taur\in \cTr}\inf_{x\in\R}\ex\left[\Phi^{*}(x + Y^{r}_{\taur}) - x\right].
$$
\end{proposition} 
Proposition will be proved in Section~\ref{proof of solution}.
Moreover 
the following important minimax result for the stopping problem 
(\ref{randomstop}) holds.
\begin{proposition}
\label{minimax}
If (\ref{Annahmen Young function}) is fulfilled, and if $\sup_{t\in [0,T]}Y_{t}\in H^{\Phi^{*}},$ then 
\begin{eqnarray*}
\sup\limits_{\tau^{r}\in\cT^{r}}\inf_{x\in\R}\ex[\Phi^{*}(x + Y_{\taur}^{r}) - x]
= 
\inf_{x\in\R}\sup\limits_{\tau^{r}\in\cT^{r}}\ex[\Phi^{*}(x + Y_{\taur}^{r}) - x].
\end{eqnarray*}
Moreover, if $(Y_{t})_{t\in [0,T]}$ is quasi-left-continuous and if $\cF_{T}$ is countably generated, then there exist $\tau^{r*}\in\cT^{r}$ and $x^{*}\in\R$ such that 
\begin{eqnarray*}
\ex[\Phi^{*}(x^{*} + Y^{r}_{\taur}) - x^{*}] 
\leq 
\ex[\Phi^{*}(x^{*} + Y^{r}_{\tau^{r*}}) - x^{*}] 
\leq 
\ex[\Phi^{*}(x + Y^{r}_{\tau^{r*}}) - x]
\end{eqnarray*}
for $x\in\R$ and $\tau\in\cT^{r}.$
\end{proposition}
The proof of Proposition \ref{minimax} can be found in Section \ref{beweis minimax}.
In the next step we shall provide conditions  ensuring that the stopping problems 
(\ref{stoppproblem}) and (\ref{randomstop}) have the same optimal value. 
\begin{proposition}
\label{derandomize2}
Let $(\Omega,\cF_{t},\pr|_{\cF_{t}})$ be atomless with countably generated $\cF_{t}$ for every $t > 0.$ If (\ref{Annahmen Young function}) is fulfilled, and if $\sup_{t\in [0,T]}Y_{t}$ belongs to $H^{\Phi^{*}},$ then 
\begin{eqnarray*}
\sup\limits_{\tau^{r}\in\cT^{r}}\inf_{x\in\R}\ex[\Phi^{*}(x + Y_{\taur}^{r}) - x]
&=&
\sup\limits_{\tau\in\cT}\inf_{x\in\R}\ex[\Phi^{*}(x + Y_{\tau}) - x]\\
&=& 
\sup\limits_{\tau\in\cT}\sup_{\Q\in\cQ_{\Phi,0}}\left(\ex_{\Q}[Y_{\tau}]- \ex\left[\Phi\left(\frac{d\pr}{d\Q}\right)\right]\right) 
\end{eqnarray*}
\end{proposition}
The proof of Proposition \ref{derandomize2} is delegated to Section \ref{beweis derandomize2}.


\section{Proofs}
\label{proofs}
We shall start with some preparations which also will turn out to be useful later on. Let us recall (cf. \cite{EdgarMilletSucheston1981}) that every 
$\taur\in\cTr$ induces a 
stochastic kernel
$
K_{\taur}:\Omega\times\cB([0,T])\rightarrow [0,1]
$ 
with $K_{\taur}(\omega,\cdot)$ being the distribution of $\taur(\omega,\cdot)$ under $\pr^{U}$ for any $\omega\in\Omega.$ Here $\cB([0,T])$ stands for the usual Borel $\sigma-$algebra on $[0,T].$ This stochastic kernel has the following properties:
\begin{eqnarray*}
&& 
K_{\taur}(\cdot,[0,t])~\mbox{is}~\cF_{t}-\mbox{measurable for every}~t\geq 0,\\
&&
K_{\taur}(\omega,[0,t]) = \sup\{u\in [0,1]\mid \taur(\omega,u)\leq t\}.
\end{eqnarray*}
The  associated stochastic kernel $K_{\taur}$ is useful to characterize the distribution function $F_{Y^{r}_{\taur}}$ of $Y_{\taur}^{r}.$ 
\begin{lemma}
\label{stopped distribution}
For any $\taur\in\cTr$ with associated stochastic kernel $K_{\taur},$ the distribution function $F_{Y_{\taur}^{r}}$ of $Y_{\taur}^{r}$ may be represented in the following way
$$
F_{Y_{\taur}^{r}}(x) = \ex[K_{\taur}(\cdot,\{t\in [0,T]\mid Y_{t}\leq x\})]\quad\mbox{for}~ x\in\R.
$$
\end{lemma}
\begin{proof}
Let $\taur\in\cTr,$ and let us fix $x\in\R.$ Then 
\begin{eqnarray*}
F_{Y_{\taur}^{r}}(x) 
= 
\ex[\eins_{]-\infty,x]}(Y_{\taur}^{r})] 
&=& 
\int_{0}^{1}\ex[\eins_{]-\infty,x]}(Y_{\taur(\cdot,u)}^{r})]\,du\\ 
&=& 
\ex\left[\int_{0}^{1}\eins_{]-\infty,x]}(Y_{\taur(\cdot,u)}^{r})\, du\right]
\end{eqnarray*}
holds (cf. \cite[Theorem 4.5]{EdgarMilletSucheston1981}), where the last equation on the right hand side is due to Fubini-Tonelli theorem. Then by definition of $K_{\taur},$ we obtain for every $\omega\in\Omega$
\begin{eqnarray*}
\int_{0}^{1}\eins_{]-\infty,x]}(Y^{r}_{\taur}(\omega,u))\, du
&=& 
\ex_{\pr^{U}}\left[\eins_{]-\infty,x]}(Y^{r}_{\taur(\omega,\cdot)}(\omega))\right]\\
&=&
\pr^{U}\left(\left\{Y^{r}_{\taur(\omega,\cdot)}(\omega)\leq x\right\}\right)\\ 
&=& 
K_{\taur}(\omega,\{t\in [0,T]\mid Y_{t}(\omega)\leq x\}). 
\end{eqnarray*}
This completes the proof.
\end{proof}
{Following a suggestion by one referee we placed the proof of Proposition \ref{solution} in front of that of Proposition \ref{minimax}.}
\subsection{Proof of Proposition \ref{solution}}
\label{proof of solution}
Let us introduce the filtered probability space $\tOFFP$ defined by 
$$
\widetilde{\cF}_{t} 
=
\bcswitch
\cF_{t}&t\leq T\\
\cF_{T}&t > T.
\ecswitch 
$$
We shall denote by $\widetilde{\cT}^{r}$ the set of randomized stopping times according to $\tOFFP.$ Furthermore, we may extend the processes 
$(Y_{t})_{t\in [0,T]}$ and $(Y^{r}_{t})_{t\in [0,T]}$ to right-continuous processes 
$(\widetilde{Y}_{t})_{t\in [0,\infty]}$ and $(\widetilde{Y}^{r}_{t})_{t\in [0,T]}$ in the following way
$$
\widetilde{Y}_{t} 
=
\bcswitch
Y_{t}&t\leq T\\
Y_{T}&t > T
\ecswitch 
\quad
\mbox{and}
\quad
\widetilde{Y}^{r}_{t} 
=
\bcswitch
Y^{r}_{t}&t\leq T\\
Y^{r}_{T}&t > T.
\ecswitch 
$$
Recall that we may equip $\widetilde{\cT}^{r}$ with the so called Baxter-Chacon topology which is compact in general, and even metrizable within our setting because $\cF_{T}$ is assumed to be countably generated (cf. Theorem 1.5 in \cite{BaxterChacon1977} and discussion afterwards). 

Next, consider the mapping
$$
\tilde{h}: \widetilde{\cT}^{r}\times\R\rightarrow\R,\, 
(\tilde{\tau}^{r},x)\mapsto \ex\left[\Phi^{*}(x + \tilde{Y}^{r}_{\tilde{\tau}^{r}}) - x\right].
$$
By assumption on $(Y_{t})_{t\in [0,T]}$, the processes $(\widetilde{Y}_{t})_{t\in [0,\infty]}$ and $(\widetilde{Y}^{r}_{t})_{t\in [0,T]}$ are quasi-left-continuous. Moreover, $\Phi^{*}$ is continuous due to Lemma \ref{optimizedcertaintyequivalent}, (i) in Appendix \ref{AppendixAA}, so that $\big(\Phi^{*}\big(x + \widetilde{Y}^{r}_{t}\big) - x\big)_{t\in [0,T]}$ is a quasi-left-continuous and right-continuous adapted process. Hence in view of \cite[Theorem 4.7]{EdgarMilletSucheston1981}, the mapping $\tilde{h}(\cdot,x)$ is continuous w.r.t. the Baxter-Chacon topology for every $x\in\R,$ and thus 
$\inf_{x\in\R}h(\cdot,x)$ is upper semicontinuous w.r.t. the Baxter-Chacon topology. Then by compactness of the Baxter-Chacon topology, we may find some randomized stopping time 
$\tilde{\tau}^{r}\in\widetilde{\cT}^{r}$ such that
$$
\inf_{x\in\R}h(\tilde{\tau}^{r},x) = \sup_{\tilde{\tau}^{r}\in\widetilde{\cT}^{r}}\inf_{x\in\R}\tilde h(\tilde{\tau}^{r},x).
$$
This completes the proof because 
$\tilde{Y}^{r}_{\tilde{\tau}^{r}} = Y^{r}_{\tilde{\tau}^{r}\wedge T}$ and 
$\tilde{\tau}^{r}\wedge T$ belongs to $\cTr$ for every $\tilde{\tau}^{r}\in\widetilde{\cT}^{r}.$ 
\hfill$\Box$


\subsection{Proof of Proposition \ref{minimax}}
\label{beweis minimax}
Let us define the mapping $h: \cTr\times\R\rightarrow\R$ by
$$
h(\taur,x) := \ex[\Phi^{*}(x + Y^{r}_{\tau^{r}}) - x].
$$
Since $\sup_{t\in [0,T]}Y_{t}$ is assumed to belong to $H^{\Phi^{*}},$ the mapping 
$\sup\limits_{\taur\in\cTr}h({\taur},\cdot)$ is finite and convex, and thus continuous. Moreover, by Lemma \ref{optimizedcertaintyequivalent} (cf. Appendix \ref{AppendixAA})
\begin{eqnarray*}
\lim_{x\to-\infty}\sup\limits_{\taur\in\cTr}h(\taur,x)\geq \lim_{x\to-\infty}(\Phi^{*}(x) - x) 
= \infty
&=&
\lim_{x\to\infty}(\Phi^{*}(x) - x)\\ 
&\leq& 
\lim_{x\to\infty} \sup\limits_{\taur\in\cTr}h(\taur,x).
\end{eqnarray*}
Hence $\inf_{x\in\R}\sup\limits_{\taur\in\cTr}h(\taur,x) = \inf_{x\in [-\varepsilon,\varepsilon]} \sup\limits_{\taur\in\cTr}h(\taur,x)$ for some $\varepsilon > 0.$ Thus $\sup\limits_{\taur\in\cTr}h(\taur,\cdot)$ attains its minimum at some $x^{*}$ due to continuity of $\sup\limits_{\taur\in\cTr}h(\taur,\cdot).$ 
Moreover, if $(Y_{t})_{t\in [0,T]}$ is quasi-left-continuous and if $\cF_{T}$ is countably generated, then $\inf_{x\in\R}h({\tau^{r*}},x) = \sup_{\taur\in\cTr}\inf_{x\in\R}h(\taur,x)$ for some $\tau^{r*}\in\cT^{r}$ due to Proposition \ref{solution}. It remains to show that 
$\sup_{\taur\in\cTr}\inf_{x\in\R}h(\taur,x) = \inf_{x\in\R}\sup_{\taur\in\cTr}h(\taur,x).$ 
Following the same line of reasoning as for the derivation of 
\eqref{DarstellungZielfunktion}, we may rewrite  $h$ in the following way.
\begin{equation}
\label{rewrite}
h(\taur,x) = \int_{0}^{\infty}\Phi^{*'}(x + z)[1 - F_{Y^{r}_{\taur}}(z)]~dz + 
\Phi^{*}(x) - x,
\end{equation}
where $F_{Y^{r}_{\taur}}$ stands for the distribution function of $Y^{r}_{\taur},$ and $\Phi^{*'}$ denotes the right-sided derivative of the convex function $\Phi^{*}.$ 
Obviously, we have
\begin{equation}
\label{convex}
h(\taur,\cdot)~\mbox{is convex and therefore continuous for every}~\taur\in\cTr.
\end{equation} 
Set $\beta := \inf\limits_{x\in\R}\sup\limits_{\tau^{r}\in\cT^{r}}\ex[\Phi^{*}(x + Y^{r}_{\tau^{r}}) - x] + 1 {= \inf\limits_{x\in\R}\sup\limits_{\tau^{r}\in\cT^{r}}h(\taur,x) + 1}$ which is a real number because {$\sup\limits_{\taur\in\cT^{r}}h(\taur,\cdot)$}
has been already proved to be a finite function {which attains its minimum on some compact interval of $\R$}. 
Furthermore, we may conclude from $h(\taur,x)\geq \Phi^{*}(x) - x$ for $x\in\R$ that
\begin{equation}
\label{infcompactneu}
I_{\beta} := \{x\in\R\mid \Phi^{*}(x) - x\leq\beta\}\,\mbox{ is a compact interval},
\end{equation}
and
\begin{equation}
\label{infcompact}
h(\tau^{r},x) > \beta\quad\mbox{for}\,\tau^{r}\in\cT^{r}\, \mbox{and}\, x\in\R\setminus I_{\beta}.
\end{equation}
By (\ref{infcompact}) we verify 
\begin{eqnarray*}
\sup_{\taur\in\cTr}\inf_{x\in\R}h(\taur,x) = \sup_{\taur\in\cTr}\inf_{x\in I_{\beta}}h(\taur,x)
\end{eqnarray*}
and
\begin{eqnarray*} 
\inf_{x\in\R}\sup_{\taur\in\cTr}h(\taur,x) = \inf_{x\in I_{\beta}}\sup_{\taur\in\cTr}h(\taur,x).
\end{eqnarray*}
{We want to apply Fan's minimax theorem (cf. \cite[Theorem 2]{Fan1953} or \cite{BorweinZhuang1986}) to $h_{\,|\,\cTr\times I_{\beta}}$. In view of \eqref{convex} and \eqref{infcompactneu}} it remains to 
show that for every $\taur_{1},\taur_{2}\in\cTr,$ and any $\lambda\in ]0,1[$ there exists some $\taur\in\cTr$ such that
\begin{eqnarray}
\label{concave} 
\lambda h(\taur_{1},x) + (1-\lambda) h(\taur_{2},x)\leq h(\taur,x)~\mbox{for all}~x\in\R. 
\end{eqnarray}
To this end let $\taur_{1},\taur_{2}\in\cTr$ with associated stochastic kernels 
$K_{\taur_{1}}, K_{\taur_{2}},$ and $\lambda\in ]0,1[.$ First, 
$K: = \lambda K_{\taur_{1}} + (1-\lambda) K_{\taur_{2}}: \Omega\times\cB([0,T])\rightarrow [0,1]$ defines a stochastic kernel satisfying
\begin{eqnarray*}
&&
K(\cdot,[0,t])~\mbox{is}~\cF_{t}-\mbox{measurable for every}~t\in [0,T],\\
&&
K(\omega,[0,T]) = 1.
\end{eqnarray*} 
Then
$$
\taur(\omega,u) := \inf\{t\in [0,T]\mid K(\omega,[0,t])\geq u\}
$$
defines some $\taur\in\cTr$ with $K_{\taur} = K.$ Furthermore, we obtain
$$
F_{Y^{r}_{\taur}} = \lambda F_{Y^{r}_{\taur_{1}}} + (1-\lambda) F_{Y^{r}_{\taur_{2}}}
$$
due to Lemma \ref{stopped distribution}. In view of (\ref{rewrite}) this implies (\ref{concave}) and the proof of Proposition \ref{minimax} is completed.
\hfill$\Box$
\subsection{Proof of Proposition \ref{derandomize2}}
\label{beweis derandomize2}
The starting idea for proving Proposition \ref{derandomize2} is to reduce the stopping problem (\ref{randomstop}) to suitably discretized random stopping times. 
The choice of the discretized randomized stopping times is suggested by the following lemma.
\begin{lemma}
\label{discretize}
For $\taur\in\cTr$ the construction
$$
\taur[j](\omega,u) := \min\{k/2^{j}\mid k\in\N, \taur(\omega,u)\leq k/2^{j}\}\wedge T
$$
defines a sequence $(\taur[j])_{j\in\N}$ in $\cTr$ satisfying the following properties.
\begin{enumerate}
\item [\rm{(i)}] $\taur[j]\searrow\taur$ pointwise, in particular it follows
$$
\lim\limits_{j\to\infty}Y^{r}_{\taur[j](\omega,u)}(\omega,u) = Y^{r}_{\taur(\omega,u)}(\omega,u)
$$
for any $\omega\in\Omega$ and every $u\in [0,1].$ 
\item [\rm{(ii)}] $\lim\limits_{j\to\infty}F_{Y^{r}_{\taur[j]}}(x) = F^{r}_{Y_{\taur}}(x)$ holds for any continuity point $x$ of $F_{Y_{\taur}}.$
\item [\rm{(iii)}] For any $x\in\R$ and every $j\in\N$ we have
$$
F_{Y^{r}_{\taur[j]}}(x) = \ex\left[\widehat{Y}_{t_{1j}}^{x} K_{\taur}(\cdot,[0,t_{1j}])\right] 
+ \sum\limits_{k=2}^{\infty}\ex\left[\widehat{Y}_{t_{kj}}^{x}\, K_{\taur}(\cdot,]t_{(k-1)j},
t_{kj}])\right],
$$
where  $t_{kj} := (k/2^{j})\wedge T$ for $k\in\N,$ and $\widehat{Y}^{x}_{t} := \eins_{]-\infty,x]}\circ Y_{t}$ for $t\in [0,T].$
\end{enumerate}
\end{lemma}
\begin{proof}
Statements (i) and (ii) are obvious, so it remains to show (iii). 
To this end recall from Lemma \ref{stopped distribution}
\begin{equation}
\label{Ausgangspunkt}
F_{Y_{\taur[j]}}(x) = \ex[K_{\taur[j]}(\cdot,\{t\in [0,T]\mid Y_{t}\leq x\})]\quad\mbox{for}~ x\in\R.
\end{equation}
Since $K_{\taur[j]}(\omega,\cdot)$ is a probability measure, we also have
\begin{eqnarray}
\label{auseinanderziehen} \nonumber
&&
K_{\taur[j]}(\omega,\{t\in [0,T]\mid Y_{t}(\omega)\leq x\})\\ \nonumber
&=& 
K_{\taur[j]}(\omega,\{t\in [0,t_{1j}]\mid Y_{t}(\omega)\leq x\})\\ \nonumber 
&&\qquad + \sum_{k = 2}^{\infty} 
K_{\taur[j]}(\omega,\{t\in ]t_{(k-1)j},t_{kj}]\mid Y_{t}(\omega)\leq x\})\\ \nonumber
&=&
K_{\taur[j]}(\omega,\{t\in [0,t_{1j}]\mid \widehat{Y}_{t}^{x}(\omega) = 1\})\\ 
&&\qquad + \sum_{k = 2}^{\infty} 
K_{\taur[j]}(\omega,\{t\in ]t_{(k-1)j},t_{kj}]\mid \widehat{Y}_{t}^{x}(\omega) = 1\})
\end{eqnarray}
for every $\omega\in\Omega.$ 
Then by definitions of $K_{\taur[j]}$ and $K_{\taur},$ 
\begin{eqnarray} \label{AnwendungKerndefinition}\nonumber
&&
K_{\taur[j]}(\omega,\{t\in ]t_{(k-1)j},t_{kj}]\mid \widehat{Y}_{t}^{x}(\omega) = 1\})
\\ \nonumber
&=&
\pr^{U}(\{\taur[j](\omega,\cdot)\in ]t_{(k-1)j},t_{kj}],\,\widehat{Y}_{\taur[j](\omega,\cdot)}^{x}(\omega) = 1\})\\ \nonumber
&=& 
\pr^{U}(\{\taur[j](\omega,\cdot) = t_{kj},\,\widehat{Y}_{t_{kj}}^{x}(\omega) = 1\})\\ \nonumber
&=&
\widehat{Y}^{x}_{t_{kj}}(\omega)\,\pr^{U}(\{\taur[j](\omega,\cdot) = t_{kj}\})\\ \nonumber
&=& 
\widehat{Y}^{x}_{t_{kj}}(\omega)\,\pr^{U}(\{\taur(\omega,\cdot)\in ]t_{(k-1)j},t_{kj}]\})\\
&=&
\widehat{Y}^{x}_{t_{kj}}(\omega)\, K_{\taur}(\omega,]t_{(k-1)j},t_{kj}])
\end{eqnarray}
for $\omega\in\Omega$ and $k\in\N$ with $k\geq 2.$ Analogously, we also obtain 
\begin{equation}
\label{AnwendungKerndefinition2}
K_{\taur[j]}(\omega,\{t\in [0,t_{1j}]\mid \widehat{Y}_{t}^{x}(\omega) = 1\}) = 
\widehat{Y}_{t_{1j}}(\omega)\, K_{\taur}(\omega,[0,t_{1j}]).
\end{equation} 
Then statement (iii) follows from (\ref{Ausgangspunkt}) combining (\ref{auseinanderziehen}) with (\ref{AnwendungKerndefinition}) and (\ref{AnwendungKerndefinition2}). The proof is finished.
\end{proof}
We shall use the discretized randomized stopping times, as defined in Lemma \ref{discretize}, to show 
that we can restrict ourselves to discrete randomized stopping times in the stopping problem (\ref{randomstop}).  
\begin{corollary}
\label{discretizedstop}
If (\ref{Annahmen Young function}) is fulfilled, then for any $\taur\in\cTr,$ we have
\begin{enumerate}
\item [(i)] 
$
\lim\limits_{j\to\infty}\ex[\Phi^{*}(x_{j} + Y^{r}_{\taur[j]}) - x_{j}] = 
\ex[\Phi^{*}(x + Y^{r}_{\taur}) - x]
$
for any sequence $(x_{j})_{j\in\N}$ in $\R^\N$ converging to some $x\in\R;$
\item [(ii)]
$\lim\limits_{j\to\infty}\inf\limits_{x\in\R}\ex[\Phi^{*}(x + Y^{r}_{\taur[j]}) - x] = 
\inf\limits_{x\in\R}\ex[\Phi^{*}(x + Y^{r}_{\taur}) - x].
$
\end{enumerate}
\end{corollary}
\begin{proof}
Let the mapping $h:\cTr\times\R$ be defined by $h(\taur,x) = \ex[\Phi^{*}(x + Y^{r}_{\taur}) - x].$ For every $\taur\in\cTr,$ the mapping $h(\taur,\cdot)$ is convex and thus continuous. Recalling that $\sup\limits_{t\geq 0} Y_{t}\in H^{\Phi^{*}}\OFP,$ a direct application of Lemma \ref{discretize}, (i), along with the dominated convergence theorem yields part (i). 
 Using terminology from \cite{RockafellarWets1998} (see also \cite{WittingMueller-Funk1995}), 
statement (i) implies that the sequence 
$(h(\taur[j],\cdot))_{j\in\N}$ of continuous mappings $h(\taur[j],\cdot)$ epi-converges to the continuous mapping $h(\taur,\cdot).$ Moreover, in view of 
(\ref{infcompactneu}) and (\ref{infcompact}), we may conclude 
$$
\lim\limits_{j\to\infty}\inf\limits_{x\in\R}h(\taur[j],x) = \inf\limits_{x\in\R}h(\taur,x),
$$
drawing on Theorem 7.31 in \cite{RockafellarWets1998} (see also Satz B 2.18 in 
\cite{WittingMueller-Funk1995}).

\end{proof}
The following result provides the remaining missing link to prove Proposition \ref{derandomize2}.
\begin{lemma}
\label{missinglink}
Let (\ref{Annahmen Young function}) be fulfilled. Furthermore, let $\taur\in\cTr,$ and let us for any $j\in\N$ denote by $\cT[j]$ the set containing all nonrandomized stopping times from $\cT$ taking  values in 
$\{(k/2^{j})\wedge T\mid k\in\N\}$ with probability \(1.\) If $(\Omega,\cF_{t},\pr|_{\cF_{t}})$ is atomless with countably generated $\cF_{t}$ for every $t > 0,$ and if $Y_{t}\in H^{\Phi^{*}}$ for $t > 0,$ then
\begin{eqnarray}
\label{tauj_minimax}
\inf_{x\in\R}\ex[\Phi^{*}(x + Y^{r}_{\taur[j]}) - x] 
\leq 
\sup_{\tau\in\cT[j]}\inf_{x\in\R}\ex[\Phi^{*}(x + Y_{\tau}) - x].
\end{eqnarray}
\end{lemma}
\begin{proof}
Let $k_{j} := \min\{k\in\N\mid k/2^{j}\geq T\}.$ If $k_{j} = 1,$ then the statement of 
Lemma \ref{missinglink} is obvious. So let us assume $k_{j}\geq 2.$ Set $t_{kj} := (k/2^{j})\wedge T$ and let the mapping $h:\cTr\times\R\rightarrow\R$ be defined via $h(\taur,x) := \ex[\Phi^{*}(x + Y_{\taur}) - x].$
We already know from Lemma \ref{discretize} that 
\begin{equation}
\label{eins}
\hspace*{-0.25cm}F_{Y^{r}_{\taur[j]}}(x) = \ex\left[\widehat{Y}_{t_{1j}}^{x} K_{\taur}(\cdot,[0,t_{1j}])\right] 
+ \sum\limits_{k=2}^{k_{j}}\ex\left[\widehat{Y}_{t_{kj}}^{x} K_{\taur}(\cdot,]t_{(k-1)j},t_{kj}])\right]
\end{equation}
holds for any $x\in\R.$ Here $\widehat{Y}^{x}_{t} := \eins_{]-\infty,x]}\circ Y_{t}$ for $t\in [0,T].$ 
Next
$$
Z_{k} := 
\bcswitch
K_{\taur}(\cdot,[0, t_{1j}])&k = 1\\
K_{\taur}(\cdot,]t_{(k-1) j},t_{kj}])&k\in\{2,...,k_{j}\}
\ecswitch
$$
defines a random variable on $(\Omega,\cF_{t_{kj}},\pr_{|\cF_{t_{kj}}})$ which satisfies 
$0\leq Z_{k}\leq 1$ $\pr-$a.s.. In addition, we may observe that $\sum_{k=1}^{k_{j}}Z_{k} = 1$ holds $\pr-$a.s.. Since the probability spaces 
$\OFPk$ $(k=1,\dots,k_{j})$ are assumed to be atomless and countably generated, we may draw on Corollary \ref{Dichtheit} (cf. Appendix \ref{AppendixC}) along with Lemma 
\ref{BanachAlaoglu} (cf. Appendix \ref{AppendixC}) and Proposition \ref{angelic} (cf. Appendix \ref{AppendixA}) to find a sequence 
$\big((B_{1n},\dots,B_{k_{j}n})\big)_{n\in\N}$ in $\Timesk\cF_{t_{kj}}$ such that $B_{1n},\dots,B_{k_{j}n}$ is a partition of $\Omega$ for $n\in\N,$ and
$$
\lim_{n\to\infty}\ex\left[\eins_{B_{kn}}\cdot g\right] = \ex\left[Z_{k}\cdot g\right]
$$
holds for $g\in L^{1}(\Omega,\cF_{t_{kj}},\pr_{|\cF_{t_{kj}}})$ and $k\in\{1,\dots,k_{j}\}.$ In particular we have by (\ref{eins}) 
$$
F_{Y^{r}_{\taur[j]}}(x) =
\lim\limits_{n\to\infty} \sum\limits_{k=1}^{k_{j}}\ex\left[\widehat{Y}_{t_{kj}}^{x} 
\eins_{B_{kn}}\right]\, \mbox{for}\, x\in\R.
$$
So by Fatou's lemma along with \eqref{rewrite},
\begin{equation}
\label{zwei}
h(\taur[j],x) \leq 
\liminf\limits_{n\to\infty} 
\int_{0}^{\infty}\Phi^{*'}(x + z)\,\Big( 1 -
\sum\limits_{k=1}^{k_{j}}\ex\left[\widehat{Y}_{t_{kj}}^{z} 
\eins_{B_{kn}}\right]\Big)\, dz + \Phi^{*}(x) - x
\end{equation}
for $x\in\R.$ Here $\Phi^{*'}$ denotes the right-sided derivative of $\Phi^{*}.$ Next we can define a sequence $(\tau_{n})_{n\in\N}$ of nonrandomized stopping times from $\cT[j]$ via
$$
\tau_{n} := \sum\limits_{k=1}^{k_{j}}t_{kj}\, \eins_{B_{kn}}.
$$
The distribution function $F_{Y_{\tau_{n}}}$ of $Y_{\tau_{n}}$ satisfies
$$
F_{Y_{\tau_{n}}}(x) = 
\sum\limits_{k=1}^{k_{j}}\ex\left[\widehat{Y}_{t_{kj}}^{x} 
\eins_{B_{kn}}\right]\, \mbox{for}\, x\in\R
$$
so that by (\ref{rewrite})
\begin{equation}
\label{drei}
h(\tau_{n},x) = \int_{0}^{\infty}\Phi^{*'}(x + z)\,
\Big( 1 -\sum\limits_{k=1}^{k_{j}}\ex\left[\widehat{Y}_{t_{kj}}^{z} 
\eins_{B_{kn}}\right]\Big)\, dz + \Phi^{*}(x) - x
\end{equation}
for $x\in\R.$ The crucial point now is to show that
\begin{description}
\item [($\star$)] $\HHH := \big\{h(\tau,\cdot)_{|I_{\beta}}\mid \tau\in\cT[j]\big\}$ is equicontinuous,
\end{description}
where $I_{\beta}$ is the interval defined in \eqref{infcompactneu}. Note that $\big(h(\tau_{n},\cdot)|_{I_{\beta}}\big)_{n\in\N}$ is a sequence in $\HHH,$ and that $\big\{h(\tau,x)\mid \tau\in\cT[j]\big\}$ is bounded for every $x\in\R.$ Thus, in view of (\ref{convex}) the statement (\(\star\)) together with Arzela-Ascoli theorem implies that  we can find a subsequence 
$\big(h(\tau_{i(n)},\cdot)|_{I_{\beta}}\big)_{n\in\N}$ such that 
$$
\lim\limits_{n\to\infty}\sup_{x\in I_{\beta}}|h(\tau_{i(n)},x) - g(x)| = 0
$$
for some continuous mapping $g: I_{\beta}\rightarrow\R.$ Hence, we may conclude from (\ref{drei}) and (\ref{zwei})
\begin{equation}
\label{vier}
g(x) = \liminf\limits_{n\to\infty} h(\tau_{i(n)},x)\geq h(\taur[j],x)\, \mbox{ for }\, x\in I_{\beta}.
\end{equation}
For any $\varepsilon > 0,$ we may find some $n_{0}\in\N$ such that 
$\sup\limits_{x\in I_{\beta}}|h(\tau_{i(n_{0})},x) - g(x)| < \varepsilon,$ which implies by (\ref{vier}) together with (\ref{infcompact}):
\begin{eqnarray*}
\inf\limits_{x\in\R}h(\tau_{i(n_{0})},x) 
\stackrel{(\ref{infcompact})}{=} 
\inf\limits_{x\in I_{\beta}}h(\tau_{i(n_{0})},x)
&\stackrel{(\ref{vier})}{\geq}& 
\inf\limits_{x\in I_{\beta}}h(\taur[j],x) - \varepsilon\\
&\geq& 
\inf\limits_{x\in\R }h(\taur[j],x) - \varepsilon
\end{eqnarray*}
and \eqref{tauj_minimax} is proved.
Therefore it remains to show the statement (\(\star\)).

\paragraph{Proof of (\(\star\))}
First, observe that for $\tau\in\cT[j]$ and real numbers $x < y,$ the 
inequality $h(\tau,x) + x \leq h(\tau,y) + y$ holds. Hence
\begin{eqnarray}
\label{fuenf}
&&
|h(\tau,x) - h(\tau,y)|\nonumber\\ 
&\leq& 
\ex\left[\Phi^{*}(y + Y_{\tau})\right] - \ex\left[\Phi^{*}(x + Y_{\tau})\right] + 
|x - y|\nonumber\\ 
&=& 
\sum\limits_{k = 1}^{k_{j}}
\ex\Big[\eins_{\{t_{kj}\}}\circ \tau\, 
\underbrace{\Big(\Phi^{*}\big(y + Y_{t_{kj}}\big) - 
\Phi^{*}\big(x + Y_{t_{kj}}\big)\Big)}_{\geq 0}\Big] + |x - y|
\nonumber\\
&\leq&
\sum\limits_{k = 1}^{k_{j}}
\ex\Big[ 
\Phi^{*}\big(y + Y_{t_{kj}}\big) - 
\Phi^{*}\big(x + Y_{t_{kj}}\big)\Big] + |x - y|
\nonumber\\
&\leq& 
\sum\limits_{k = 1}^{k_{j}}|h(t_{kj},x) - h(t_{kj},y)| + (k_{j} + 1)\, |x-y|
\end{eqnarray} 
By convexity, the mappings $h(t_{kj},\cdot),$ $k=1,...,k_{j},$ are also locally Lipschitz continuous. Thus, in view of (\ref{fuenf}),  it is easy to verify that $\HHH$ is equicontinuous at every $x\in I_{\beta}.$ This proves  (\(\star\)).
\end{proof}
Now, we are ready to prove Proposition \ref{derandomize2}. By (\ref{Hilfsstoppproblem}) we have 
\begin{eqnarray*}
\sup\limits_{\tau^{r}\in\cT^{r}}\inf_{x\in\R}\ex[\Phi^{*}(x + Y^{r}_{\taur}) - x]
&\geq&
\sup\limits_{\tau\in\cT}\inf_{x\in\R}\ex[\Phi^{*}(x + Y_{\tau}) - x]\\
&=& 
\sup\limits_{\tau\in\cT}\sup_{\Q\in\cQ_{\Phi,0}}\left(\ex_{\Q}[Y_{\tau}]- \ex\left[\Phi\left(\frac{d\pr}{d\Q}\right)\right]\right) 
\end{eqnarray*}
Moreover, due to (ii) of Corollary \ref{discretizedstop} and Lemma \ref{missinglink} we  conclude that for any 
$\taur\in\cTr$ 
\begin{eqnarray*}
\inf_{x\in\R}\ex[\Phi^{*}(x + Y^{r}_{\taur}) - x] 
&=& 
\lim_{j\to\infty}\inf_{x\in\R}\ex[\Phi^{*}(x + Y^{r}_{\taur[j]}) - x]\\ 
&\leq& 
\sup_{\tau\in\cT}\inf_{x\in\R}\ex[\Phi^{*}(x + Y_{\tau}) - x].
\end{eqnarray*}
Thus  Proposition \ref{derandomize2} is proved.
\hfill$\Box$

 \subsection{Proof of Theorem \ref{new_representation}}
\label{beweis new dual representation}
First, we get from  Propositions \ref{minimax} and \ref{derandomize2} 
\begin{eqnarray*}
\inf_{x\in\R}\sup\limits_{\tau^{r}\in\cT^{r}}\ex[\Phi^{*}(x + Y^{r}_{\tau^r}) - x]
&=& 
\sup\limits_{\tau^{r}\in\cT^{r}}\inf_{x\in\R}\ex[\Phi^{*}(x + Y^{r}_{{\taur}}) - x]\\
&=&
\sup\limits_{\tau\in\cT}\inf_{x\in\R}\ex[\Phi^{*}(x + Y_{\tau}) - x]. 
\end{eqnarray*}
Furthermore, 
\begin{eqnarray*}
\inf_{x\in\R}\sup\limits_{\tau^{r}\in\cT^{r}}\ex[\Phi^{*}(x + Y^{r}_{\tau^r}) - x]
&\geq& 
\inf_{x\in\R}\sup\limits_{\tau\in\cT}\ex[\Phi^{*}(x + Y_{\tau}) - x]\\
&\geq& 
\sup\limits_{\tau\in\cT}\inf_{x\in\R}\ex[\Phi^{*}(x + Y_{\tau}) - x].
\end{eqnarray*}
Thus 
$$
\sup\limits_{\tau\in\cT}\inf_{x\in\R}\ex[\Phi^{*}(x + Y_{\tau}) - x] = 
\inf_{x\in\R}\sup\limits_{\tau\in\cT}\ex[\Phi^{*}(x + Y_{\tau}) - x]
$$
which completes the proof of Theorem \ref{new_representation}. 
\hfill$\Box$

\subsection{Proof of Proposition \ref{saddle-point}}
\label{proof of saddle-point}
Just simplifying notation, we assume that $\TTT = \{0,1,\dots,T\}$ with $T$ being a positive integer.   
By \eqref{discreteminimax} we have 
$$
\sup\limits_{\tau\in\cT_{\TTT}}\inf_{x\in\R}\ex[\Phi^{*}(x + Y_{\tau}) - x] = \inf_{x\in\R}\sup\limits_{\tau\in\cT_{\TTT}}\ex[\Phi^{*}(x + Y_{\tau}) - x].
$$
So it is left to show that there exists a solution $\tau^{*}$ of the maximization problem \eqref{primalproblem} and a solution $x^{*}$ of the minimization problem \eqref{dualproblem}. Indeed such a pair $(\tau^{*},x^{*})$  would be as required.
\medskip
 
In view of \eqref{infcompact}, we may find some compact interval $I$ of $\R$ such that 
\begin{equation}
\label{Hilfsproblem}
\sup_{\tau\in\cT_{\mathbb{T}}}\inf_{x\in\R}\ex\left[\Phi^{*}(x + Y_{\tau}) - x\right] = 
\sup_{\tau\in\cT_{\mathbb{T}}}\inf_{x\in I}\ex\left[\Phi^{*}(x + Y_{\tau}) - x\right].
\end{equation}
Let $\cC(I)$ denote the space of continuous real-valued mappings on $I.$ This space will be equipped with the sup-norm $\|\cdot\|_{\infty},$ whereas the product $\cC(I)^{T}$ is viewed to be endowed with the norm $\|\cdot\|_{\infty,T},$ defined by 
$
\|(f_{1},\dots,f_{T})\|_{\infty,T}:= \sum_{t=1}^{T}\|f_{t}\|_{\infty}.
$
The key in solving the maximization problem \eqref{primalproblem} is to show that 
\begin{equation}
\label{MengeK}
K := \left\{(G_{1,A_{1}},\dots G_{T,A_{n}})\mid (A_{1},\dots,A_{T})\in\cP_{T}\right\}
\end{equation} 
is a weakly compact subset of $\cC(I)^{T}$ w.r.t. the norm $\|\cdot\|_{\infty,T}.$
Here $\cP_{T}$ stands for the set of all $(A_{1},\dots,A_{T})$ satisfying $A_{t}\in\cF_{t}$ for $t\in\{1,\dots,T\}$ as well as $\pr(A_{t}\cap A_{s}) = 0$ for $t\neq s,$ and 
$\pr(\cup_{t = 1}^{T} A_{t}) = 1.$ Furthermore, define
$$
G_{t,A_{t}}: I\rightarrow\R,~x\mapsto \ex\left[\eins_{A_{t}}\cdot
\left(\Phi^{*}(x + Y_{t}) - x\right)\right]\quad\mbox{for}~t\in\{1,\dots,T\},~A_{t}\in\cF_{t}.
$$ 
Notice that any mapping $G_{t,A_{t}}$ is extendable to a real-valued convex function on \(\R\), and therefore also continuous.
\par
Before proceeding, we need some further notation, namely $\cP_{T}^{\infty}$ denoting the set of all $(f_{1},\dots,f_{T})$ satisfying $f_{t}\in L^{\infty}(\Omega,\cF_{t},\pr|_{\cF_{t}})$ with $f_{t}\geq 0~\pr-$a.s. for $t\in\{1,\dots,T\},$ and $\sum_{t=1}^{T}f_{t} = 1~\pr-$a.s.. Obviously, the subset $\{(\eins_{A_{1}},\dots,\eins_{A_{T}})\mid (A_{1},\dots,A_{T})\in\cP_{T}\}$ consists of extreme points of $\cP_{T}^{\infty}.$ Any 
$f_{t}\in L^{\infty}(\Omega,\cF_{t},\pr|_{\cF_{t}})$ may be associated with the mapping
$$
H_{t,f_{t}}: I\rightarrow\R,~x\mapsto \ex\left[f_{t}\cdot\left(\Phi^{*}(x + Y_{t}) - x\right)\right]\quad(t \in\{1,\dots,T\}).
$$
It is extendable to a real-valued convex function on \(\R\), and thus also continuous. Hence, the mapping 
$$
\Lambda: \TimestT L^{\infty}(\Omega,\cF_{t},\pr|_{\cF_{t}})\rightarrow \cC(I)^{T},~
(f_{1},\dots,f_{T})\mapsto (H_{1,f_{1}},\dots,H_{T,f_{T}})
$$
is well-defined, and obviously linear. In addition it satisfies the following convenient continuity property.
\begin{lemma}
\label{KompaktheitPartI}
Let $\TimestT \sigma(L^{\infty}_{t},L^{1}_{t})$ be the product topology of 
$\sigma(L^{\infty}_{t},L^{1}_{t})$ $(t = 1,\dots,T)$ on $\TimestT L^{\infty}(\Omega,\cF_{t},\pr|_{\cF_{t}}),$
where $\sigma(L^{\infty}_{t},L^{1}_{t})$ denotes the weak* topology on 
$L^{\infty}(\Omega,\cF_{t},\pr|_{\cF_{t}}).$ 

Then, $\cP_{T}^{\infty}$ is compact w.r.t. $\TimestT \sigma(L^{\infty}_{t},L^{1}_{t}),$ and the mapping $\Lambda$ is continuous w.r.t. $\TimestT \sigma(L^{\infty}_{t},L^{1}_{t})$ and the weak topology induced by $\|\cdot\|_{\infty,T}.$ In particular 
the image $\Lambda(\cP_{T}^{\infty})$ of $\cP_{T}^{\infty}$ under $\Lambda$ is weakly compact w.r.t. $\|\cdot\|_{\infty,T}.$
\end{lemma}
\begin{proof}
The continuity of $\Lambda$ follows in nearly the same way as in the proof of Proposition 3.1 from \cite{Edwards1987}. Moreover, $\cP_{T}^{\infty}$ is obviously closed w.r.t. the product topology $\TimestT \sigma(L^{\infty}_{t},L^{1}_{t}),$ and even compact due to Banach-Alaoglu theorem. Then by continuity of $\Lambda,$ the set $\Lambda(\cP_{T}^{\infty})$ is weakly compact w.r.t. $\|\cdot\|_{\infty,T}.$ This completes the proof.
\end{proof}
We need some further preparation to utilize Lemma \ref{KompaktheitPartI}.
\begin{lemma}
\label{thin}
Let $s,t\in\{1,\dots,T\}$ with $t\leq s,$ and let $A\in\cF_{T}.$ If $(\Omega,\cF_{t},\pr|_{\cF_{t}})$ is atomless and if $\left\{\ex\left[\eins_{A}\cdot\Phi^{*}(x + Y_{s})~|~\cF_{t}\right]\mid x\in\R\right\}$ is a thin subset of $L^{1}(\Omega,\cF_{t},\pr|_{\cF_{t}})$, then 
$\left\{\ex\left[\eins_{A}\cdot\left(\Phi^{*}(x + Y_{s}) - x\right)~|~\cF_{t}\right]\mid x\in\R\right\}$ is a thin subset of $L^{1}(\Omega,\cF_{t},\pr|_{\cF_{t}}).$
\end{lemma}
\begin{proof}
Let $A\in\cF_{t}$ with $\pr(A) > 0.$ Since $(\Omega,\cF_{t},\pr|_{\cF_{t}})$ is atomless, we may find disjoint $B_{1}, B_{2}\in\cF_{t}$ contained in $A$ with $\pr(B_{1}), \pr(B_{2}) > 0.$ Then by assumption there exist nonzero $f_{1}, f_{2}\in L^{\infty}(\Omega,\cF_{t},\pr|_{\cF_{t}})$ such that $f_{i}$ vanishes outside $B_{i}$ as well as 
$\ex\left[~f_{i}\cdot\ex\left[\eins_{A}\cdot\Phi^{*}(x + Y_{s})~|~\cF_{t}\right]~\right] = 0$ for $x\in\R$ and $i\in\{1,2\}.$

Moreover, we may choose $\lambda_{1},\lambda_{2}\in\R$ with $\lambda_{i}\not= 0$ for at least one $i\in\{1,2\}$ and 
$\ex\left[\left(\lambda_{1}f_{1} + \lambda_{2} f_{2}\right)\cdot\eins_{A}\right] = 0.$ Finally, $\lambda_{1}f_{1} + \lambda_{2} f_{2}\in L^{\infty}(\Omega,\cF_{t},\pr|_{\cF_{t}})\setminus\{0\},$ and, setting $f := \lambda_{1}f_{1} + \lambda_{2} f_{2},$ 
\begin{eqnarray*}
&&
\ex\left[~f\cdot\ex\left[\eins_{A}\cdot\left(\Phi^{*}(x + Y_{s}) - x\right)~|~\cF_{t}\right]~\right]\\ 
&=& 
\sum_{i=1}^{2}\lambda_{i}~\ex\left[~f_{i}\cdot\ex\left[\eins_{A}\cdot\Phi^{*}(x + Y_{s}) ~|~\cF_{t}\right]~\right] 
-
x~\ex\left[\left(\lambda_{1} f_{1} + \lambda_{2} f_{2}\right)\cdot\eins_{A}\right] 
= 
0
\end{eqnarray*}
for $x\in\R.$ This completes the proof.
\end{proof}
The missing link in concluding the desired compactness of the set $K$ from 
\eqref{MengeK} 
is provided by the following auxiliary result.
\begin{lemma}
\label{Schlussel}
Let $(\Omega,\cF_{t},\pr|_{\cF_{t}})$ be atomless for $t\in\{1,\dots,T\},$ and furthermore let 
the subset $\left\{\ex\left[\eins_{A}\cdot\Phi^{*}(x + Y_{s})~|~\cF_{t}\right]\mid x\in\R\right\}$ of $L^{1}(\Omega,\cF_{t},\pr|_{\cF_{t}})$ be thin for arbitrary $s,t\in\{1,\dots,T\}$ with $t\leq s$ and $A\in\cF_{T}.$ 

Then for any $(f_{1},\dots,f_{T})\in\cP^{\infty}_{T},$ there exist $(A_{1},\dots,A_{T})\in\cP_{T}$ and mappings $g_{t}\in L^{\infty}(\Omega,\cF_{t},\pr|_{\cF_{t}})$ 
$(t=1,\dots,T)$ such that $\Lambda(g_{1},\dots,g_{T}) \equiv 0,$ and 
$$
(f_{1},\dots,f_{T}) = (\eins_{A_{1}},\dots,\eins_{A_{T}}) + (g_{1},\dots, g_{T})\quad\pr-\mbox{a.s.}.
$$
\end{lemma}
\begin{proof}
Let $s,t\in\{1,\dots,T\}$ with $t\leq s$ and $A\in\cF_{T}.$  We may draw on Lemma \ref{thin} to observe that $\{\ex\left[\eins_{A}\cdot(\Phi^{*}(x + Y_{s}) - x)|\cF_{t}\right]\mid x\in\R\}$ is a thin subset of $L^{1}(\Omega,\cF_{t},\pr|_{\cF_{t}}).$ Then the statement of Lemma \ref{Schlussel} follows immediately from Proposition \ref{KreinMilman} (cf. Appendix \ref{AppendixC}) applied to the sets $M_{t}$ ($t = 1,\dots,T$), where 
$M_{t}:= \{\Phi^{*}(x + Y_{t}) - x\mid x\in\R\}.$ 
\end{proof}
Under the assumptions of Lemma \ref{Schlussel}, the set $K$ defined in \eqref{MengeK} coincides with $\Lambda(\cP_{T}^{\infty}),$ which in turn is weakly compact w.r.t. $\|\cdot\|_{\infty,T}$ due to Lemma \ref{KompaktheitPartI}.
\begin{corollary}
\label{KompaktheitPartII}
Under the assumptions of Lemma \ref{Schlussel}, the set $K$ (cf. \eqref{MengeK})  is weakly compact w.r.t. $\|\cdot\|_{\infty,T}.$
\end{corollary}
Now we are ready to select a solution of the maximization problem \eqref{primalproblem}. 

\medskip

\noindent
\underline{Existence of a solution of maximization problem \eqref{primalproblem}}:\\[0.2cm]
Let the assumptions of Proposition \ref{saddle-point} be fulfilled. In view of \eqref{Hilfsproblem} it suffices to solve
$$
\mbox{maximize}\quad \inf_{x\in I}\ex[\Phi^{*}(x + Y_{\tau}) - x]\,\mbox{over}\,\tau\in\cT_{\TTT}.
$$
Let us assume that 
$\sup_{\tau\in\cT_{\TTT}}\inf_{x\in I}\ex[\Phi^{*}(x + Y_{\tau}) - x] > 
\inf_{x\in I}\ex[\Phi^{*}(x + Y_{0}) - x]$ because otherwise $\tau\equiv 0$ would be optimal. Since $\pr(A)\in\{0,1\}$ for $A\in\cF_{0}$ by assumption, any stopping time 
$\tau\in\TTT\setminus\{0\}$ is concentrated on $\{1,\dots,T\}.$

By Corollary \ref{KompaktheitPartII}, the set $K$ (cf. \eqref{MengeK}) is weakly compact w.r.t. the norm $\|\cdot\|_{\infty,T}.$ Furthermore, the concave mapping 
$L:\cC(I)^{T}\rightarrow\R,$ defined by $L(r_{1},\dots,r_{T}) :=  \inf_{x\in I} \sum_{t = 1}^{T} r_{t}(x),$ is continuous w.r.t. $\|\cdot\|_{\infty,T}.$ This means that $-L$ is convex as well as well as $\|\cdot\|_{\infty,T}-$continuous, and thus also weakly lower semicontinuous because $\|\cdot\|_{\infty,T}-$closed convex subsets are also weakly closed. Hence $L$ is weakly upper semicontinuous, and therefore its restriction to $K$ attains a maximum. In particular, the set
$$
\left\{\inf_{x\in I}\ex\left[\Phi^{*}(x + Y_{\tau}) - x\right]\mid \tau\in\cT_{\TTT}\setminus\{0\}\right\} = L(K)
$$
has a maximum. This shows that we may find a solution of \eqref{primalproblem}.\hfill$\Box$
\medskip

\noindent
\underline{Existence of a solution of problem \eqref{dualproblem}}:\\[0.2cm]
By $l(x) := \sup_{\tau\in\cT_{\TTT}}\ex[\Phi^{*}(x + Y_{\tau}) - x]$ we may define a convex, and therefore also continuous mapping $l:\R\rightarrow\R.$ Moreover by Lemma 
\ref{optimizedcertaintyequivalent} (cf. Appendix \ref{AppendixAA}),
$$
\lim_{x\to-\infty}l(x)\geq\lim_{x\to-\infty}\big(\Phi^{*}(x) - x\big) = \infty = 
 \lim_{x\to\infty}\big(\Phi^{*}(x) - x\big)\leq \lim_{x\to\infty}l(x).
$$
This means that $\inf_{x\in\R}l(x) = \inf_{x\in [-\varepsilon,\varepsilon]}l(x)$ for some $\varepsilon > 0.$ Hence $l$ attains its minimum at some $x^{*}\in [-\varepsilon,\varepsilon]$ because $l$ is continuous. Any such $x^{*}$ is a solution of the problem \eqref{dualproblem}.
\hfill$\Box$

\begin{appendix}
\section{Appendix}
\label{AppendixAA}
\begin{lemma}
\label{optimizedcertaintyequivalent}
Let $\Phi: [0,\infty[\rightarrow [0,\infty]$ be a lower semicontinuous, convex mapping satisfying $\inf_{x\geq 0}\Phi(x) = 0,$ and  
$\lim_{x\to\infty} \frac{\Phi(x)}{x} = \infty.$ Furthermore, let 
$\cQ_{\Phi,0}$ denote the set of all probability measures $\Q$ on $\cF$ which are absolutely continuous w.r.t. $\pr$ such that the Radon-Nikodym derivative $\frac{d\Q}{d\pr}$ satisfies $\ex\left[\Phi\left(\frac{d\Q}{d\pr}\right)\right] < \infty.$ Then the following statements hold true.
\begin{enumerate}
\item [(i)]
If $\Phi(x_{0}) < \infty$ for some $x_{0} > 0,$ then the Fenchel-Legendre transform $\Phi^{*}:\R\rightarrow\R\cup\{\infty\}$ of $\Phi$ is a nondecreasing, convex finite mapping. In particular its restriction $\Phi^{*}\bigl |_{[0,\infty[}$ to 
$[0,\infty[$ is a finite Young-function, which in addition satisfies the condition
$\lim_{x\to\infty}(\Phi^{*}(x) - x) = \infty$ if $x_{0} > 1,$ and 
$\lim_{x\to-\infty}(\Phi^{*}(x) - x) = \infty$ in the case of $x_{0} < 1.$
\item [(ii)] 
If $\Phi(x_{0}),\Phi(x_{1}) < \infty$ for some $x_{0} < 1 < x_{1},$ then for any $X$ from $H^{\Phi^{*}} := H^{\Phi^{*}\bigl |_{[0,\infty[}},$ we obtain
$$
\sup_{\Q\in\cQ_{\Phi,0}}\left(\ex_{\Q}[X]- \ex\left[\Phi\left(\frac{d\Q}{d\pr}\right)\right]\right) = \inf_{x\in\R}\ex[\Phi^{*}(x + X) - x],
$$
where the supremum on the left hand side of the equality is attained for some 
$\Q\in\cQ_{\Phi,0}.$
\end{enumerate}
\end{lemma}
\begin{proof}
Let \(\Phi(x_0)<\infty\) for some \(x_0>0.\) Obviously, $\Phi^{*}$ is a nondecreasing convex function satisfying the properties 
\begin{equation}
\label{young}
\Phi^{*}(0) = -\inf_{x\geq 0}\Phi(x) = 0\,\,\mbox{and}\, 
\lim_{y\to\infty}\Phi^{*}(y)\geq \lim_{y\to\infty}(x_{0}y - \Phi(x_{0})) = \infty.
\end{equation}
Next, we want to verify the finiteness of $\Phi^{*}.$ Since $\Phi^{*}$ is {nondecreasing}, and $\Phi^{*}(y)\geq x_{0}y - \Phi(x_{0}) > - \infty$ holds for any $y\in\R,$ it suffices to show that $\Phi^{*}(y) < \infty$ for every $y\geq 0.$ For that purpose consider the mapping
$$
\beta: [0,\infty[\times [0,\infty[\rightarrow [-\infty,\infty[,\, (y,x)\mapsto xy - \Phi(x). 
$$
By assumption on $\Phi,$ we have 
$$
\lim_{x\to\infty} \beta(y,x) = 
\lim_{x\to\infty} x \left(y -  \frac{\Phi(x)}{x}\right) = - \infty < \beta(y,x_{0})\quad \mbox{for}\, y \geq 0.
$$ 
Hence for any $y\geq 0,$ we may find some $z_{y}\in [x_{0},\infty[$ such that we obtain 
$\Phi^{*}(y) = \sup_{0\leq x\leq z_{y}}\beta(y,x).$ Moreover,  
$\beta(y,\cdot)$ is upper semicontinuous for $y\geq 0.$ Hence, for every $y\geq 0,$ there is some $x\in [0,z_{y}]$ with  $\Phi^{*}(y) = \sup_{0\leq x\leq z_{y}}\beta(y,x) = \beta(y,x) < \infty.$

As a finite convex function $\Phi^{*}$ is continuous. Since it is also nondecreasing, we may conclude from \eqref{young} that its restriction to $[0,\infty[$ is a finite Young function. Let us now assume that $x_{0} > 1.$ Then 
$$
\lim_{y\to\infty}(\Phi^{*}(y) - y)\geq \lim_{y\to\infty}\big((x_{0} - 1) y - \Phi(x_{0})\big) = \infty.
$$
Analogously, $\lim_{y\to-\infty}(\Phi^{*}(y) - y) = \infty$ may be derived in the case of $x_{0} < 1.$ Thus we have proved the full statement (i).

Let us turn over to the proof of statement (ii), and let us consider the mapping
$$
\rho: H^{\Phi^{*}}\rightarrow [-\infty,\infty[,\, X\mapsto \inf_{x\in\R}\ex[\Phi^{*}(x - X) - x]
$$
Then, due to convexity of $\Phi^{*},$ we may apply Jensen's inequality along with statement (i) to conclude 
$$
\lim_{x\to - \infty}\ex[\Phi^{*}(x - X) - x]\geq 
\lim_{x\to - \infty}([\Phi^{*}(x - \ex[X])]-x) = \infty\quad\mbox{for}\, X\in H^{\Phi^{*}},
$$ 
and
$$
\lim_{x\to \infty}\ex[\Phi^{*}(x - X) - x]\geq 
\lim_{x\to \infty}[\Phi^{*}(x - \ex[X]) - x] = \infty\quad\mbox{for}\, X\in H^{\Phi^{*}}.
$$
Thus, for any $X\in H^{\Phi^{*}},$ we find some $\delta_{X} > 0$ such that 
$$
\rho(X) = \inf_{x\in [-\delta_{X},\delta_{X}]}\ex[\Phi^{*}(x - X) - x].
$$ 
In addition, for $X\in H^{\Phi^{*}},$ the mapping $x\mapsto \ex[\Phi^{*}(x - X) - x]$ is a convex mapping on $\R,$ hence its restriction to $[-\delta_{X},\delta_{X}]$ is continuous. This implies that $\rho$ is a real-valued function. 

Moreover, it is easy to check that $\rho$ is a so called convex risk measure, defined to mean that it satisfies the following properties.
\begin{itemize}
    \item monotonicity: $\rho(X)\ge\rho(Y)$ for all $X,Y\in H^{\Phi^{*}}$ with $X\le Y$,
    \item cash-invariance: $\rho(X+m)=\rho(X)-m$ for all $X\in H^{\Phi^{*}}$ and $m\in\R$,
    \item convexity: $\rho(\lambda X + (1-\lambda) Y)\le\lambda\rho(X)+ (1-\lambda)\rho(Y)$ for all $X,Y\in H^{\Phi^{*}},$ $\lambda\in [0,1].$
\end{itemize}
Then we obtain from Theorem 4.3 in \cite{CheriditoLi2009} that 
$$
\rho(X) = \max_{\Q\in\cQ_{\Phi,0}} \left(\ex_{\Q}[-X]- \rho^{*}(\Q)\right)
$$
holds for all $X\in H^{\Phi^{*}},$ where
$$
\rho^{*}(\Q) := \sup_{X\in H^{\Phi^{*}}}\left(\ex_{\Q}[-X] - \rho(X)\right).
$$
By routine procedures we may verify
$$
\rho^{*}(\Q) = \sup_{X\in H^{\Phi^{*}}}\left(\ex_{\Q}[X] - \ex[\Phi^{*}(X)]\right)
$$
for $\Q\in\cQ_{\Phi,0}.$ Since $\lim_{x\to -\infty}[\Phi^{*}(x) - x] = \lim_{x\to \infty}[\Phi^{*}(x) - x] = \infty$ due to statement (i), we may conclude from 
\cite[(5.23)]{CheriditoLi2009}
$$
\Phi^{*}(\Q) = \sup_{X\in H^{\Phi^{*}}}\left(\ex_{\Q}[X] - \ex[\rho^{*}(X)]\right) = 
\ex\left[\Phi\left(\frac{d\Q}{d\pr}\right)\right]\quad\mbox{for all}\, \Q\in\cQ_{\Phi,0}.
$$
This completes the proof.
\end{proof}

\section{Appendix}
\label{AppendixA}
Let $(\oOmega,\ocF,(\ocF)_{i\in\{1,\dots,m\}},\oP)$ be a filtered probability space, and let the product space 
$\Timesim L^{\infty}(\oOmega,\ocF_{i},\oP|_{\ocF_{i}})$ be endowed with the product topology 
$\Timesim\sigma(L^{\infty}_{i},L^{1}_{i})$ of the weak* topologies 
$\sigma(L^{\infty}_{i},L^{1}_{i})$ on $L^{\infty}(\oOmega,\ocF_{i},\oP|_{\ocF_{i}})$ 
(for $i=1,\dots,m$).
\begin{proposition}
\label{angelic}
Let $L^{1}(\oOmega,\ocF_{i},\oP|_{\ocF_{i}})$ be separable w.r.t. the weak topology $\sigma(L^{1}_{i},L^{\infty}_{i})$ for $i\in\{1,\dots,m\},$ and let $\cA\subseteq\Timesim L^{\infty}(\oOmega,\ocF_{i},\oP|_{\ocF_{i}})$ be relatively compact w.r.t. $\Timesim\sigma(L^{\infty}_{i},L^{1}_{i}).$\par 
Then for any $X$ from the 
$\Timesim\sigma(L^{\infty}_{i},L^{1}_{i})-$closure of $\cA,$ we may find a sequence $(X_{n})_{n\in\N}$ in $\cA$ which converges to $X$ w.r.t. the $\Timesim\sigma(L^{\infty}_{i},L^{1}_{i}).$
\end{proposition}
\begin{proof}
Setting $E := \Timesim L^{\infty}(\oOmega,\ocF_{i},\oP|_{\ocF_{i}}),$ we shall denote by $E'$ the topological dual of $E$ w.r.t. $\Timesim\sigma(L^{\infty}_{i},L^{1}_{i}).$ It is easy to check that $$
\Gamma(g_{1},\dots,g_{m})(f_{1},\dots f_{m}) := 
\sum_{i=1}^{m}\ex\left[f_{i}\cdot g_{i}\right],
$$ 
where $g_{i}\in L^{1}(\oOmega,\ocF_{i},\oP|_{\ocF_{i}})$ and 
$f_{i}\in L^{\infty}(\oOmega,\ocF_{i},\oP|_{\ocF_{i}})$ (for $i=1,\dots,m$) defines 
a linear operator from $\Timesim L^{1}(\oOmega,\ocF_{i},\oP|_{\ocF_{i}})$ onto $E'$ which is continuous w.r.t. the product topology $\Timesim\sigma(L^{1}_{i},L^{\infty}_{i})$ of the weak topologies $\sigma(L^{1}_{1},L^{\infty}_{1}),\dots,\sigma(L^{1}_{m},L^{\infty}_{m})$ and the weak topology $\sigma(E',E).$\par  
Since $\Timesim\sigma(L^{1}_{i},L^{\infty}_{i})$ is separable by assumption, we may conclude that $\sigma(E',E)$ is separable too. Then the statement of the Proposition 
\ref{angelic} follows immediately from \cite{Floret1978}, p.30.
\end{proof}
\section{Appendix}
\label{AppendixC}
Let for $m\in\N$ denote by $(\oOmega,\ocF,(\ocF_{i})_{i\in\{1,\dots,m\}},\oP)$ a filtered probability space, and let the set $\overline{\cP}_{m}$ gather all sets
$(A_{1},\dots,A_{m})$ from $\Timesim\ocF_{i}$ satisfying $\oP(A_{i}\cap A_{j}) = 0$ for 
$i\not= j$ and $\oP(\bigcup_{i=1}^{m} A_{i}) = 1.$ We shall endow respectively the product spaces 
$\Timesikm L^{\infty}(\oOmega,\ocF_{i},\oP|_{\ocF_{i}})$ with the product topologies 
$\Timesikm\sigma(L^{\infty}_{i},L^{1}_{i})$ of the weak* topologies 
$\sigma(L^{\infty}_{i},L^{1}_{i})$ on $L^{\infty}(\oOmega,\ocF_{i},\oP|_{\ocF_{i}})$ 
(for $k\in\{1,\dots,m\}$ and $i=k,\dots,m$). Fixing $k\in\{1,\dots,m\}$ and 
nonnegative $h\in L^{\infty}(\oOmega,\ocF_{k},\oP|_{\ocF_{k}}),$ the subset 
$\ocPinfty_{mk}(h)\subseteq \Timesikm L^{\infty}(\oOmega,\ocF_{i},\oP|_{\ocF_{i}})$ is defined to consist of all $(f_{k},\dots,f_{m})\in \Timesikm L^{\infty}(\oOmega,\ocF_{i},\oP|_{\ocF_{i}})$ such that $f_{i}\geq 0$ $\oP-$a.s. for any $i\in\{k,\dots,m\}$ and 
$\sum_{i=k}^{m}f_{i} = h$ $\oP-$a.s.. For abbreviation we shall use notation 
$\ocPinfty_{m} := \ocPinfty_{m1}(1).$
\begin{lemma}
\label{BanachAlaoglu}
$\ocPinfty_{mk}(h)$ is a compact subset of $\Timesikm L^{\infty}(\oOmega,\ocF_{i},\oP|_{\ocF_{i}})$ w.r.t. the topology $\Timesikm\sigma(L^{\infty}_{i},L^{1}_{i})$ for 
$k\in\{1,\dots,m\}$ and arbitrary nonnegative $h\in L^{\infty}(\oOmega,\ocF_{k},\oP|_{\ocF_{k}}).$
\end{lemma}
\begin{proof}
The statement of Lemma \ref{BanachAlaoglu} is obvious in view of the Banach-Alaoglu theorem.
\end{proof}
\begin{proposition}
\label{KreinMilmanallgemein}
Let $M_{i}\subseteq L^{1}(\oOmega,\ocF_{i},\oP|_{\ocF_{i}})$ be nonvoid for $i = 1,\dots,m$ such that $\{\ex\left[\eins_{A}\cdot f~|~\ocF_{i}\right]\mid f\in M_{j}\}$ is a thin subset of $L^{1}(\oOmega,\ocF_{i},\oP|_{\ocF_{i}})$ for $i,j\in\{1,\dots,m\}$ with $i\leq j$ and any $A\in\ocF_{m}.$ Furthermore, let us fix $(f_{1},\dots,f_{m})\in\ocPinfty_{m}$ and consider the set $N_{1}$ consisting of all 
$(h_{1},\dots,h_{m})$ from $\Timesim L^{\infty}(\oOmega,\ocF_{i},\oP|_{\ocF_{i}})$ satisfying $\ex\left[h_{i}\cdot \varphi_{i}\right] = \ex\left[f_{i}\cdot \varphi_{i}\right]$ for 
any $\varphi_{i}\in M_{i},$ $i=1,\dots,m.$
Then the set $N_{1}\cap~\ocPinfty_{m}$ has extreme points, and for each extreme point 
$(h^{*}_{1},\dots,h^{*}_{m}),$ there exists some 
$(A_{1},\dots,A_{m})\in\overline{\cP}_{m}$ such that 
$h^*_{i} = \eins_{A_{i}}$ $\oP-$ a.s. holds for $i = 1,\dots,m.$
\end{proposition}
\begin{proof}
We shall use ideas from the proof of Proposition 6 in \cite{KuehnRoesler1998}.
\medskip

First, let us, for any $k\in\{1,\dots,m\}$, denote by $N_{k}$ the set of all $(h_{k},\dots,h_{m})$ from $\Timesikm L^{\infty}(\oOmega,\ocF_{i},\oP|_{\ocF_{i}})$ satisfying $\ex\left[h_{i}\cdot \varphi_{i}\right] = \ex\left[f_{i}\cdot \varphi_{i}\right]$ for 
$\varphi_{i}\in M_{i}$ and $i=k,\dots,m.$ It is 
closed w.r.t. $\Timesikm\sigma(L^{\infty}_{i},L^{1}_{i}).$ Hence by Lemma \ref{BanachAlaoglu}, the set 
$K_{k}(h) := N_{k}\cap~\ocPinfty_{mk}(h)$ is compact w.r.t. $\Timesikm\sigma(L^{\infty}_{i},L^{1}_{i})$ for every nonnegative $h\in L^{\infty}(\oOmega,\ocF_{k},\oP|_{\ocF_{k}}).$ 
Since it is also convex, we may use the Krein-Milman theorem to conclude that each set $K_{k}(h)$ has some extreme point if it is nonvoid. Notice that $K_{1}(1)$ contains at least $(f_{1},\dots,f_{m})$ so that it has some extreme point. We shall now show by backward induction that for any $k\in\{1,\dots,m\}$ and any nonnegative 
$h\in L^{\infty}(\oOmega,\ocF_{k},\oP|_{\ocF_{k}})$ with nonvoid $K_{k}(h)$ 
\begin{itemize}
\item[$(\star\star)$]
each of its extreme points $(h^{*}_{k},\dots,h^{*}_{m})$ satisfies $h^{*}_{i} = h\cdot\eins_{A_{i}}$ $\oP-$a.s. 
($i=k,\dots,m$) for some $(A_{1},\dots,A_{m})\in\overline{\cP}_{m}$ with 
$A_{i} = \emptyset$ if $i < k.$ 
\end{itemize}
Obviously, this would imply the statement of Proposition \ref{KreinMilmanallgemein}.
\medskip

For $k = m,$ the set $K_{m}(h)$ is nonvoid iff 
$\ex\left[h\cdot\varphi_{m}\right] = \ex\left[f_{m}\cdot\varphi_{m}\right]$ holds for every $\varphi_{m}\in M_{m}.$ In this case, $h$ is the only extreme point, which has trivial representation $h = h\cdot\eins_{\Omega}$ 
corresponding to $(\emptyset,\dots,\emptyset,\Omega)\in\overline{\cP}_{m}.$ 
\medskip

Now let us assume that for some $k\in\{2,\dots,m\}$ and every nonvoid $K_{k}(h)$ statement $(\star\star)$ is satisfied. Let $h\in L^{\infty}(\oOmega,\ocF_{k-1},\oP|_{\ocF_{k-1}})$ be nonnegative with $K_{k-1}(h)\not=\emptyset,$ and select any extreme point 
$(h^*_{k-1},\dots,h^*_{m})$ of $K_{k-1}(h).$ Then $h - h^*_{k-1}$ belongs to 
$L^{\infty}(\oOmega,\ocF_{k-1},\oP|_{\ocF_{k-1}})$ and is nonnegative. Moreover, 
$(h^*_{k},\dots,h^*_{m})\in K_{k}(h - h^*_{k-1}),$ and it is easy to check that 
$(h^*_{k},\dots,h^*_{m})$ is even an extreme point of $K_{k}(h - h^*_{k-1}).$ Hence by assumption, there exists some $(A_{1},\dots,A_{m})\in\overline{\cP}_{m}$ satisfying 
$A_{i} = \emptyset$ if $i\leq k - 1$ and $h^*_{i} = (h - h^*_{k-1})\cdot \eins_{A_{i}}$ $\oP-$a.s. for $i = k,\dots,m.$
\medskip

Setting $D := \{h^*_{k-1} > 0\}\cap \{h - h^*_{k-1} > 0\},$ we want to show $\oP(D) = 0.$ This will be done by contradiction assuming $\oP(D) > 0.$ Then $\oP(D_{\varepsilon}) > 0$ for some $\varepsilon > 0,$ where $D_{\varepsilon} := \{h^*_{k-1} > \varepsilon\}\cap \{h - h^*_{k-1} > \varepsilon\}.$\par
We may observe by assumption that  
$\{\ex\left[\eins_{A_{i}}\cdot\varphi_{i}~|~\ocF_{k-1}\right]\mid \varphi_{i} \in M_{i}\}$ 
(with $i=k,\dots,m$) as well as $M_{k-1}$ are all thin subsets of $L^{1}(\oOmega,\ocF_{k-1},\oP|_{\ocF_{k-1}}).$ Since finite unions of thin subsets are thin subsets again (cf. \cite[Proposition 2.1]{Anantharaman2012}), we may find some nonzero 
$g\in L^{\infty}(\oOmega,\ocF_{k-1},\oP|_{\ocF_{k-1}})$ vanishing outside $D_{\varepsilon},$ and satisfying $\ex\left[g\cdot\varphi_{k-1}\right] = 0$ for 
$\varphi_{k-1} \in M_{k-1}$ as well as 
$$
\ex\left[g\cdot\eins_{A_{i}}\cdot\varphi_{i}\right] 
= 
\ex\left[g\cdot\ex\left[\eins_{A_{i}}\cdot\varphi_{i}~|~\ocF_{k-1}\right]\right] 
= 
0
\quad (\varphi_{i}\in M_{i},~i\in\{k,\dots,m\}).
$$
According to Theorem 2.4 in \cite{Anantharaman2012}, we may choose $g$ such that 
$$
\oP(\{|g| = 1\}\cap D_{\varepsilon}) = \oP(D_{\varepsilon})
$$ 
holds. Now, define $(\widehat{h}_{k-1},\dots,\widehat{h}_{m})$ and $(\overline{h}_{k-1},\dots,\overline{h}_{m})$ by
$$
\widehat{h}_{i} := 
\bcswitch
h^*_{i} + \varepsilon~g~& i = k-1\\
h^*_{i} - \varepsilon~g~\eins_{A_{i}}& \mbox{otherwise}
\ecswitch
~\mbox{and}~~    
\overline{h}_{i} := 
\bcswitch
h^*_{i} - \varepsilon~g~& i = k - 1\\
h^*_{i} + \varepsilon~g~\eins_{A_{i}}&\mbox{otherwise}
\ecswitch.
$$
Since $\oP(A_{i}\cap A_{j}) = 0$ for $i\not= j$ and $\oP(\bigcup_{i=k}^{m}A_{i}) = 1,$ we obtain $\sum_{i=k}^{m}g\cdot \eins_{A_{i}} = g$ $\oP-$a.s.. So by construction,  
$(\widehat{h}_{1},\dots,\widehat{h}_{m}), (\overline{h}_{1},\dots,\overline{h}_{m})$ differ, and belong both to $K_{k-1}(h).$ Moreover, $h^*_{i} = \widehat{h}_{i}/2 + \overline{h}_{i}/2$ for 
$i = k-1,\dots,m.$ This contradicts the fact that $(h^*_{k-1},\dots,h^*_{m})$ is an extreme point of $K_{k-1}(h).$ Therefore, $\oP(D) = 0.$\par 
Now define $(B_{1},\dots,B_{m})\in\Timesim\ocF_{i}$ by
$$
B_{i} := 
\bcswitch
\{h^*_{k-1} > 0, h = h^*_{k-1}\}& i = k-1\\
A_{i}\cap \{h^*_{k-1} = 0\}& i\in\{k,\dots,m\}\\
\emptyset&\mbox{otherwise}.
\ecswitch
$$
Obviously, $\oP(B_{i}\cap B_{j}) = 0$ for $i\not= j$ follows from $\oP(A_{i}\cap A_{j}) = 0$ for $i\not= j.$ Moreover, $\oP(\bigcup_{i=1}^{m}B_{i})\geq 
\oP(\Omega\setminus D~\cap~ \bigcup_{i=k}^{m}A_{i}) = 1.$ In particular 
$(B_{1},\dots,B_{m})\in\overline{\cP}_{m}.$ Finally, it may be verified easily that $h^*_{i} = h\cdot\eins_{B_{i}}$ $\oP-$a.s. holds for $i= k-1,\dots,m.$ Hence $K_{k-1}(h)$ fulfills statement $(\star\star)$ completing the proof.
\end{proof}
\begin{proposition}
\label{KreinMilman}
Let $M_{i}\subseteq L^{1}(\oOmega,\ocF_{i},\oP|_{\ocF_{i}})$ be nonvoid for $i = 1,\dots,m$ such that $\{\ex\left[\eins_{A}\cdot f~|~\ocF_{i}\right]\mid f\in M_{j}\}$ is a thin subset of $L^{1}(\oOmega,\ocF_{i},\oP|_{\ocF_{i}})$ for $i,j\in\{1,\dots,m\}$ with $i\leq j$ and any $A\in\ocF_{m}.$\par
Then for any $(f_{1},\dots,f_{m})\in\ocPinfty_{m},$ there exist  
$(A_{1},\dots,A_{m})\in\overline{\cP}_{m}$ and $g_{i}\in L^{\infty}(\oOmega,\ocF_{i},\oP|_{\ocF_{i}})$ $(i=1,\dots,m)$ such that 
$$
\ex\left[g_{i}\cdot \varphi_{i}\right] = 0\quad\mbox{for}~\varphi_{i}\in M_{i}~\mbox{with}~i = 1,\dots,m,
$$ 
and
$$
(f_{1},\dots,f_{m}) = (\eins_{A_{1}},\dots,\eins_{A_{m}}) + (g_{1},\dots,g_{m})\quad\oP-\mbox{a.s..}
$$
\end{proposition}
\begin{proof}
Let us fix any $(f_{1},\dots,f_{T})\in\ocPinfty_{m},$ and let $N_{1}$ denote the set consisting of all 
$(h_{1},\dots,h_{m}),$ where $h_{i}\in L^{\infty}(\oOmega,\ocF_{i},\pr|_{\ocF_{i}})$ such that $\ex\left[h_{i}\cdot \varphi_{i}\right] = \ex\left[f_{i}\cdot \varphi_{i}\right]$ for $\varphi_{i}\in M_{i}.$ By Proposition \ref{KreinMilmanallgemein}, we may select an extreme point $(h_{1},\dots,h_{m})$ of $N_{1}\cap~\ocPinfty_{m}$ and some 
$(A_{1},\dots,A_{m})\in\overline{\cP}_{m}$ such that $h_{i} = \eins_{A_{i}}$ $\oP-$a.s. holds for $i=1,\dots,m.$ Then $(g_{1},\dots,g_{m}) := (f_{1} - h_{1},\dots,f_{m} - h_{m})$ and $(A_{1},\dots,A_{m})$ are as required.
\end{proof}
\begin{corollary}
\label{Dichtheit}
If $(\oOmega,\ocF_{i},\oP|_{\ocF_{i}})$ is atomless for every $i\in\{1,\dots,m\},$ then $\ocPinfty_{m}$ is the $\Timesim\sigma(L^{\infty}_{i},L^{1}_{i})-$closure of
$$
\{(\eins_{A_{1}},\dots,\eins_{A_{m}})\mid (A_{1},\dots,A_{m})\in\overline{\cP}_{m}\}.
$$
\end{corollary}
\begin{proof}
Let $(f_{1},\dots,f_{m})\in\ocPinfty_{m}$ be arbitrary. Consider the subsets 
$$
U_{i\varepsilon}(M_{i}) := \{\varphi\in L^{\infty}(\oOmega,\ocF_{i},\oP|_{\ocF_{i}})\mid 
\big|\ex\left[(f_{i} - \varphi)\cdot f\big|\right] < \varepsilon~\mbox{for}~f\in M_{i}\},
$$
where $\varepsilon > 0,$ and $M_{i}$ any nonvoid, finite subset of $L^{1}(\oOmega,\ocF_{i},\oP|_{\ocF_{i}}).$ The sets $\Timesim U_{i\varepsilon}(M_{i})$ constitute a basis of the $\Timesim\sigma(L^{\infty}_{i},L^{1}_{i})-$neighbourhoods of $(f_{1},\dots,f_{m}).$ So let us select any $\varepsilon > 0$ and nonvoid finite subsets $M_{i}$ of 
$L^{1}(\oOmega,\ocF_{i},\oP|_{\ocF_{i}})$ for $i=1,\dots,m.$\par

Let $i,j\in\{1,\dots,m\}$ with $i\leq j,$ and $A\in\ocF_{m}.$ Then the set 
consisting of all $\ex\left[\eins_{A}\cdot f~|~\ocF_{i}\right]$ with $f\in M_{j}$ is a nonvoid finite subset of $L^{1}(\oOmega,\ocF_{i},\oP|_{\ocF_{i}}),$ in particular it is 
thin because $(\oOmega,\ocF_{i},\oP|_{\ocF_{i}})$ is assumed to be atomless (cf. 
\cite[Lemma 2]{KingmanRobertson1968}). Hence we may apply Proposition \ref{KreinMilman} to select some $(A_{1},\dots,A_{m})\in\overline{\cP}_{m}$ satisfying 
$\ex\left[(f_{i} - \eins_{A_{i}})\cdot f\right] = 0$ for $f\in M_{i}$ and $i\in\{1,\dots,m\}.$ This means 
$$
(\eins_{A_{1}},\dots,\eins_{A_{m}})\in\Timesim U_{i\varepsilon}(M_{i}),
$$
and completes the proof.
\end{proof}
\end{appendix}

\section*{Acknowledgements}
The authors would like to thank Alexander Schied and Mikhail Urusov for fruitful discussions and helpful remarks.

\end{document}